\documentclass[11pt,letterpaper]{article}
\usepackage[T1]{fontenc}
\usepackage{newtxtext,newtxmath}
\usepackage[letterpaper,left=0.85in,right=0.85in,top=0.8in,bottom=0.8in,
  headheight=14pt,headsep=0.18in,footskip=0.38in]{geometry}
\usepackage{microtype}
\usepackage{amsfonts}
\usepackage{amsmath}
\usepackage{mathtools}
\usepackage{graphicx}
\usepackage{xcolor}
\usepackage{booktabs}
\usepackage{bm}
\usepackage{comment}
\usepackage{authblk}
\usepackage{caption}
\usepackage[numbers,sort&compress]{natbib}
\usepackage{hyperref}
\usepackage{cleveref}

\usepackage{xurl}

\graphicspath{{figs/}}
\DeclareGraphicsExtensions{.pdf,.png,.jpg}

\hypersetup{
  pdftitle={Multifidelity Formulations for Triangular Transport},
  pdfauthor={O. Davis, D. Sharp, Y. M. Marzouk, and G. Geraci},
  colorlinks=true,
  linkcolor=blue!50!black,
  citecolor=blue!50!black,
  urlcolor=blue!50!black,
  hypertexnames=false
}

\newtheorem{theorem}{Theorem}[section]

\newtheorem{proposition}[theorem]{Proposition}

\newenvironment{proof}{\par\noindent\textit{Proof.}\ }{\hfill$\square$\par}
\crefname{hypothesis}{Hypothesis}{Hypotheses}
\crefname{fact}{Fact}{Facts}

\newenvironment{keywords}{\par\smallskip\noindent\textbf{Keywords:}\ }{\par}
\newenvironment{MSCcodes}{\par\smallskip\noindent\textbf{MSC Codes:}\ }{\par}

\newcommand{\DGS}[1]{}
\newcommand{\DGSi}[1]{}
\newcommand{\OND}[1]{}
\newcommand{\ONDi}[1]{}
\newcommand{\todo}[2][]{}

\newcommand{\SZero}{\bm{S}_0} %
\newcommand{\SCR}{\widehat{\bm{S}}^{\mathrm{H\textrm{-}CR}}} %
\newcommand{\SFR}{\widehat{\bm{S}}^{\mathrm{H\textrm{-}FR}}} %
\newcommand{\SNH}{\widehat{\bm{S}}^{\mathrm{NH}}} %
\newcommand{\SSF}{\widehat{\bm{S}}^{\mathrm{SF}}} %

\title{Multifidelity Formulations for Triangular Transport}
\author[1]{Owen Davis\thanks{\href{mailto:ondavis@sandia.gov}{ondavis@sandia.gov} (Corresponding author)}}
\author[1,2]{Daniel Sharp\thanks{\href{mailto:dannys4@mit.edu}{dannys4@mit.edu}}}
\author[2]{Youssef Marzouk\thanks{\href{mailto:ymarz@mit.edu}{ymarz@mit.edu}}}
\author[1]{Gianluca Geraci\thanks{\href{mailto:ggeraci@sandia.gov}{ggeraci@sandia.gov}}}
\affil[1]{%
Department of Uncertainty Quantification and Optimization,\protect\\%
Sandia National Laboratories, Albuquerque, NM%
\vspace{1em}
}
\affil[2]{%
Center for Computational Science and Engineering,\protect\\%
Massachusetts Institute of Technology, Cambridge, MA%
}
\date{}

\begin{document}

\maketitle
\begin{abstract}
We develop multifidelity methods for constructing triangular transport maps from samples, when high-fidelity data are scarce but lower-fidelity data are more abundant. Using this set of multifidelity data, we approximate a triangular transport map that bijectively maps between a tractable reference density and the \textit{high-fidelity} target distribution. We introduce two strategies to leverage low-fidelity data: a hierarchical approach that composes maps between adjacent fidelity levels, and a non-hierarchical method that incorporates low-fidelity information through monotonicity-preserving corrections to the map parameterization. Numerical experiments compare these strategies with single-fidelity transport and demonstrate how the proposed multifidelity approaches can improve map estimation from limited high-fidelity data. To illustrate the broader utility of the learned maps, we also deploy them in a downstream amortized simulation-based inference task. This example shows that multifidelity improvements in map estimation can translate to improved conditional sampling and uncertainty quantification when high-fidelity data are scarce. %
\end{abstract}



\section{Introduction}\label{sec:introduction}

Many methods for %
Bayesian computation and uncertainty quantification (UQ) sample from a desired complex target distribution by coupling this target with a well-understood reference distribution, via a deterministic map. 
When this map %
is explicitly learned, it can be used for density evaluation, sample generation, or statistical inference.%

This work focuses on \emph{triangular transportation of measure}, realized by transport maps with a ``triangular'' structure that offers many computational and theoretical advantages~\cite{villani2008optimal,carlier2010knothe}.
Triangular maps have been used in Bayesian inference~\cite{Marzouk2016,el2012bayesian,rubio2023transport,baptista2024bayesian,parno2016multiscale,brennan2020greedy}, optimal experimental design~\cite{huan2024optimal,koval2024tractable,li2024expected}, density estimation~\cite{wang2022minimax,martinez2024decomposing,lopez2024density}, generative modeling~\cite{irons2022triangular,katzfuss2024scalable}, and data assimilation~\cite{spantini2018inference,ramgraber2023ensemblea,ramgraber2023ensembleb,grange2024distributed}. 
Broader introductions can be found in~\cite{Marzouk2016,baptista2024representation,ramgraber2025friendly}.

In this article, we focus on learning triangular transport maps from data, i.e., from a set of samples drawn from the target distribution, without any direct knowledge of the target density.
A practical obstacle in this setting is that accurate map construction can require many target samples. When datasets are expensive to acquire---as is often the case for expensive simulations and physical experiments---it may be challenging to accurately learn these transport maps. %
Here, we address the problem of scarce data by using \emph{multifidelity} information: an ensemble of related information sources with different costs and accuracies. %

Specifically, we assume access to a small set of samples from the high-fidelity target distribution and more plentiful samples from an ensemble of related lower-fidelity distributions. We use these multifidelity data to estimate a triangular transport map coupling the high-fidelity target to a tractable reference distribution, with the goal of improving upon single-fidelity estimation from the scarce high-fidelity samples alone. The resulting map provides a tractable representation of the high-fidelity target, enabling efficient sampling, density evaluation, and even simulation-based inference. We refer to the construction of triangular transport maps from multifidelity data as \textit{multifidelity triangular transport}, and we study the distinct function approximation challenges that arise in this setting.

Multifidelity methods are common in query-intensive tasks~\cite{peherstorfer2018survey}, including Monte Carlo estimation~\cite{Gorodetsky2020CV,Haji2016,Giles2008,peherstorfer2016optimal}, surrogate modeling~\cite{kennedy2000predicting,Perdikaris2015,Le2014,Gorodetsky2020CV,GorodetskyIJUQ,ZengCMAME,zeng2025boosting}, likelihood-based inference~\cite{peherstorfer2019transport,muchandimath2025synchronized}, and likelihood-free inference~\cite{krouglova2025multifidelity,prescott2020multifidelity,prescott2021multifidelity}.
Many multifidelity methods are bi-fidelity, involving one high-fidelity and one low-fidelity data source. Extensions to more than two sources often impose a fidelity hierarchy and apply bi-fidelity corrections between adjacent levels. This approach is sensitive to the choice %
of hierarchy: if there is an uninformative or corrupted source, errors can propagate and degrade the overall approximation. %

More general multifidelity surrogate methods instead model relationships %
among information sources with a prescribed directed acyclic graph, where \textit{hierarchical} and \textit{peer} %
structures are special cases~\cite{GorodetskyCM,GorodetskyIJUQ,ZengCMAME,zeng2025boosting}.
In a \textit{hierarchical} structure, each model is assumed to be less informative, with respect to the high-fidelity model, than the model directly above it on the hierarchy. This produces a chain in which lower-fidelity information influences the high-fidelity approximation indirectly through intermediate models. In a \textit{peer} structure, there is no assumed relationship among the low-fidelity models, e.g., all models could be equally informative with respect to the high-fidelity model. %
The chosen dependency structure strongly affects performance (unless mitigation strategies are introduced; see, e.g., \cite{zeng2025boosting}), and peer structures are generally more robust to uninformative or corrupted low-fidelity sources. %

Motivated by this perspective, we propose two families of multifidelity triangular transport approaches. The first is a hierarchical multifidelity transport construction that trains a sequence of triangular maps between adjacent fidelity levels and composes them to reach the high-fidelity target.
Within this family, we consider two formulations that differ in how the hierarchy is used to construct the intermediate transport problems, leading to fixed-reference and changing-reference variants with distinct optimization properties. %
The second family introduces a non-hierarchical construction that uses all low-fidelity sources simultaneously by introducing corrections into the triangular map parameterization. In this setting, no fidelity hierarchy is required and the parametric corrections allow us to automatically %
suppress misleading lower-fidelity information. %
This approach corresponds to a peer-structured dependency graph in which all low-fidelity sources are direct parents of the high-fidelity source. %

We demonstrate, evaluate, and compare both families of constructions for several distribution learning tasks and for a canonical simulation-based inference problem, adapted to the multifidelity setting.

The remainder of the paper is organized as follows: Section~\ref{sec:relation_to_other_work} discusses the technical background and places our contributions within the literature. Section~\ref{sec:triangular_transport} reviews notation and the sample-based transport learning objective used throughout the paper. Sections~\ref{sec:hierarchical_multifidelity} and~\ref{sec:nonhierarchical_multifidelity} introduce hierarchical and non-hierarchical multifidelity transport methods, respectively. %
Section~\ref{sec:conditions_of_effectiveness} discusses the advantages and disadvantages of our new approaches. Section~\ref{sec:numerics} numerically explores our multifidelity methods and compares them to single-fidelity approaches that rely only on high-fidelity data. Section~\ref{sec:conclusion} offers conclusions and directions for future work. 

\section{Related work}\label{sec:relation_to_other_work}

In this section, we review existing work that relates to our setting of learning from samples drawn from distributions of differing fidelity. We highlight, however, that few of these use the entire multifidelity dataset when learning a transport map itself, and fewer still allow for sampling from conditional distributions---a key feature of triangular transport.

The works \cite{krouglova2025multifidelity, prescott2020multifidelity, prescott2021multifidelity} develop multifidelity methods for likelihood-free\slash simulation-based inference. These methods %
aim to characterize posterior distributions when simulation information is available at multiple fidelities. They do not, however, develop triangular transport constructions. In particular, \cite{prescott2020multifidelity, prescott2021multifidelity} develop multifidelity approximate Bayesian computation methods, while \cite{krouglova2025multifidelity} uses transfer learning to incorporate low-fidelity data into a neural density estimator. These approaches largely address bi-fidelity settings, though \cite{krouglova2025multifidelity} describes an untested extension to additional fidelities.%

The works \cite{muchandimath2025synchronized, muchandimath2025accelerating, peherstorfer2019transport} use multifidelity information in likelihood-based Bayesian inference. In \cite{muchandimath2025synchronized}, the authors introduce a Markov chain Monte Carlo (MCMC) method, SYNCE, that couples chains across fidelity levels using common random numbers; \cite{muchandimath2025accelerating} then uses triangular transport maps to accelerate this sampling strategy. Similarly, \cite{peherstorfer2019transport} constructs a transport map from a low-fidelity unnormalized density and uses it to guide proposals in a high-fidelity Metropolis--Hastings algorithm. These works combine multifidelity information with sampling and transport maps, where the transport maps are used to accelerate or guide likelihood-based inference. They do not address the problem considered here: how to train a high-fidelity triangular transport map using samples from target distributions of varying fidelity.

Two recent works are more closely related. \cite{calle2025generative} develops a Bayesian triangular transport model for multiscale spatial fields, using scale-aware autoregressive Gaussian processes to represent fine-scale distributions conditioned on coarse-scale inputs. Their setting is primarily one of spatial downscaling and multiscale conditional generation from paired fields across resolutions. In contrast, the present work is not restricted to spatial fields, resolution hierarchies, or conditional generation from coarse-fine data. %
\cite{sharp2026sampling} develops a gradient-free Bayesian inference framework combining annealing, multifidelity models, importance-weighted quadrature, and transport-based posterior surrogates. %
In contrast to the Bayesian inference setting of~\cite{sharp2026sampling}, the present work assumes that samples at each fidelity are already available and studies how those samples can be used directly to estimate a high-fidelity triangular transport map. Thus, while~\cite{sharp2026sampling} uses transport within a multifidelity setting, the focus here is on the design of transport maps themselves for a fixed multifidelity problem. Further, this work also examines the case when no hierarchy exists.

\section{Triangular transport}\label{sec:triangular_transport}
Consider target density $\pi$ and reference density $\eta$ on $\mathbb{R}^d$.
For a monotone triangular transport map $\bm{S}:\mathbb{R}^d\to\mathbb{R}^d$, the $k$th output is a function $S_k:\mathbb{R}^k\to\mathbb{R}$ depending on the first $k$ coordinates that is monotone increasing in the last input:%

\begin{equation}\label{eqn:tri_transport_structure}
    [\bm{S}(x_1,\ldots, x_d)]_k = S_k(x_1,\ldots,x_k) %
    \qquad
    \partial_k S_k \coloneq \frac{\partial}{\partial x_k} S_k(x_1,\ldots, x_k) > 0,\quad
    k=1,\ldots,d.
\end{equation}

The Jacobian $\nabla\bm{S}$ is therefore lower triangular and positive definite, and we denote with $\mathcal{F}_{\triangle}$ the set of all monotone triangular maps $\bm{S}:\mathbb{R}^d\rightarrow \mathbb{R}^d$ satisfying \eqref{eqn:tri_transport_structure}. Monotone triangular structure ensures that these maps are invertible.

We seek maps that \textit{push forward} the target density $\pi$ to the reference density $\eta$, which we denote by $\bm{S}_{\sharp}\pi=\eta$; we will also use the notation $\bm{S}^{\sharp}\eta \coloneq (\bm{S}^{-1})_\sharp \eta = \pi$ to denote $\bm{S}$ \textit{pulling back} $\eta$ to $\pi$. 
Accordingly, given a random variable $X\sim\pi$, we have $\bm{S}(X)\sim\eta$; conversely, if $Z\sim\eta$, then $\bm{S}^{-1}(Z)\sim\pi$. %
When it is easy to sample from $\eta$, $\bm{S}$ allows us to sample from $\pi$ via the inverse map; thus, this setting is thus often referred to as \textit{inverse transport}. The corresponding density transformation is:

\begin{equation}\label{eqn:density_pullback}
    \bigl(\bm{S}^\sharp \eta\bigr) (x)    
    =
    \eta(\bm{S}(\bm{x}))
    \left|\det\nabla\bm{S}(\bm{x})\right|.
\end{equation}

Under standard absolute continuity assumptions, there exists a unique map $\bm{S}\in\mathcal{F}_{\triangle}$ coupling $\pi$ and $\eta$; this map is the \textit{Knothe--Rosenblatt rearrangement} or KR map~\cite{rosenblatt1952remarks,carlier2010knothe,bonnotte2013knothe,Marzouk2016}.
Inverting the KR map is computationally convenient: to evaluate $\bm{S}^{-1}$, we solve $d$ one-dimensional root-finding problems (cf.\ Appendix~\ref{app:triangular_transport_details}).

\subsection{Learning a transport map from samples}\label{sec:inverse_transport}

In our setting, the target distribution $\pi$ is known only through samples $\{\bm{x}^{(j)}\}_{j=1}^{N}\sim\pi$. The goal is to approximate the map $\bm{S}^*\in\mathcal{F}_{\triangle}$ satisfying $\bm{S}^*_{\sharp}\pi=\eta$. Following~\cite{Marzouk2016}, this map $\bm{S}^*$
satisfies
\begin{equation}\label{eqn:tri_map_optimization_KL}
    \bm{S}^*\in\arg\min_{\bm{S}\in\mathcal{F}_{\triangle}}
    \mathcal{D}_{\mathrm{KL}}\!\bigl(
        \pi\, \| \,\bm{S}^{\sharp}\eta
    \bigr),
\end{equation}
where $\mathcal{D}_{\mathrm{KL}}$ is the Kullback--Leibler (KL) divergence. This optimization problem is equivalent, up to constants independent of $\bm{S}$, to solving

\begin{equation}\label{eqn:tri_map_optimization_step2}
    \min_{\bm{S}\in\mathcal{F}_{\triangle}}
    \mathbb{E}_{\pi}
    \left[
        -\log \eta(\bm{S}(\bm{x}))
        -\log \det\nabla\bm{S}(\bm{x})
    \right].
\end{equation}

When $\eta$ is log-concave, the optimization problem~\eqref{eqn:tri_map_optimization_step2} is convex in the map $\bm{S}$ \cite{opt_convexity}. The corresponding sample-average approximation, which can also be understood as a maximum likelihood estimation problem \cite{wang2022minimax}, is:

\begin{equation}\label{eqn:tri_map_optimization_saa}
    \min_{\bm{S}\in\mathcal{F}_{\triangle}}
    \frac{1}{N}
    \sum_{j=1}^{N}
    \left[
        -\log \eta\!\left(\bm{S}(\bm{x}^{(j)})\right)
        -\log \det\nabla\bm{S}(\bm{x}^{(j)})
    \right].
\end{equation}

The reference density is typically a freely chosen tractable distribution; throughout this work, we use the standard Gaussian, a common log-concave choice for which the objective separates across triangular components:

\begin{equation}\label{eqn:tri_map_kth_component_optimization}
    \min_{
    S_k:\mathbb{R}^k\to\mathbb{R},\ \partial_k S>0
    }\ %
    \frac{1}{N}
    \sum_{j=1}^{N}
    \left[
        \frac{1}{2}S_k^2(\bm{x}^{(j)})
        -\log \partial_k S_k(\bm{x}^{(j)})
    \right],
    \  k=1,\ldots,d.%
\end{equation}
This separation is computationally advantageous because the map components can be trained in parallel. 

\subsection{Map component parameterization}
\label{sec:component_parameterization}

We enforce the monotonicity constraint of $\mathcal{F}_{\triangle}$ using
an integrated-rectifier parameterization~\cite{baptista2024representation}. We write $\bm{S}=\mathcal{R}(\bm f)$, where  $\bm{f}=(f_1,\ldots, f_d)$ is lower-triangular and $f_k:\mathbb{R}^k\rightarrow \mathbb{R}$. The operator $\mathcal{R}$ acts componentwise as%

\begin{equation}\label{eqn:integrated_rectifier}
    S_k(\bm{x}_{\leq k}) = \mathcal{R}_k(f_k)(\bm{x}_{\leq k})
    =
    \underbrace{
    f_k(\bm{x}_{\leq k-1},0)
    }_{\text{non-monotone}}
    +
    \underbrace{
    \int_0^{x_k}
    g\!\left(
        \partial_k f_k(\bm{x}_{\leq k-1},t)
    \right)\,dt
    }_{\text{monotone}},
    \qquad k=1,\ldots,d,
\end{equation}
where $g:\mathbb{R}\to(0,\infty)$ is a positive, bijective rectifier. Monotonicity in the final coordinate is thus enforced by construction:%

\begin{equation}\label{eqn:integrated_rectifier_derivative}
    \partial_k S_k(\bm{x}_{\leq k})
    =
    g\!\left(
        \partial_k f_k(\bm{x}_{\leq k})
    \right)
    >0.
\end{equation}

For log-concave $\eta$, the infinite-dimensional transport objective \eqref{eqn:tri_map_optimization_step2} and its sample-average approximation \eqref{eqn:tri_map_optimization_saa} are convex in $\bm{S}$. For the standard Gaussian reference used here, the latter separates into the componentwise optimization problems \eqref{eqn:tri_map_kth_component_optimization}, which are convex in $S_k$. Under the reparameterization $S_k = \mathcal{R}_k(f_k)$, however, these objectives are generally non-convex as a functional of $f_k$. Nevertheless, \cite{baptista2024representation} establishes conditions under which this infinite-dimensional optimization over $f_k$ has no spurious local minima. Their analysis motivates our choice $g:=\mathrm{SoftPlus}$, which satisfies the rectifier-specific conditions imposed therein.

For computation, we restrict each $f_k$ to a finite-dimensional space $V_k^p$ spanned by multivariate Hermite functions of total order at most $p$, and write $f_k = f_k(\cdot;\bm{\theta}_k)$ to make explicit the coefficients $\bm{\theta}_k$. 
Further details are provided in Appendix~\ref{app:map_param} and~\cite{baptista2024representation,ramgraber2025friendly}.

\section{Triangular transport from multifidelity information}
\label{sec:multifidelity_inverse_transport}

We now consider sample-based triangular transport in a multifidelity setting. Our aim is to learn a transport map coupling a high-fidelity target distribution to a tractable
reference distribution using samples from the high-fidelity target and from several related 
lower-fidelity distributions.   
Throughout this section, fidelity is indexed by $\ell=0,\ldots,M$, where $\ell=0$ denotes high-fidelity and $\ell>0$ %
denotes lower-fidelity. %
Component indices are written $k=1,\ldots,d$, and sample indices are denoted by superscripts. Thus, for target density $\pi_\ell$ at fidelity $\ell$, %
we have a dataset $\mathcal{X}_\ell = \{\bm{x}_\ell^{(j)}\}_{j=1}^{N_\ell}$, with $\bm{x}_\ell^{(j)}\sim\pi_\ell$.%

We summarize here the key assumptions of our framework: %
First, the high-fidelity density $\pi_0$ is more expensive to sample than the lower-fidelity densities $\pi_\ell$, $\ell=1,\ldots,M$, i.e., %
$N_0\ll N_\ell$. Second, the lower-fidelities ($\ell\geq 1)$ are unordered a priori. %
We wish to use the small %
high-fidelity dataset $\mathcal{X}_0$ in concert with the low-fidelity datasets $\{\mathcal{X}_\ell\}_{\ell=1}^{M}$ to approximate the high-fidelity map $\bm{S}_0$, defined as the monotone triangular map achieving:

\begin{equation}\label{eqn:hf_kr_map}
    (\SZero)_{\sharp}\pi_0=\eta.
\end{equation}

A central concern is how multifidelity information should be leveraged to improve upon a single-fidelity approximation of $\SZero$. Existing theoretical results provide useful motivation for our proposed multifidelity formulations through a common complexity-reduction principle.
First, approximation-theoretic results for inverse triangular transport connect error in the map to error in the associated pullback measure~\cite{baptista2025approximation}. Since \eqref{eqn:hf_kr_map} equivalently gives $\pi_0 = (\SZero)^{\sharp}\eta$, an approximate map $\widehat{\bm S}_0$ induces the target measure approximation $\widehat{\pi}_0 =(\widehat{\bm S}_0)^{\sharp}\eta$. Theorem~4.10 of~\cite{baptista2025approximation} bounds the KL divergence between pullback measures $\bm F^{\sharp}\eta$ and $\bm G^{\sharp}\eta$ by an appropriate function-space error between the maps $\bm F$ and $\bm G$, while Theorem~8.3 of~\cite{baptista2025approximation} shows that, over a chosen approximation class, the smallest achievable KL divergence between $\pi_0$ and a pullback $\bm S^{\sharp}\eta$ is controlled by how well that class can approximate $\SZero$. Thus, controlling the function approximation error in $\SZero$ provides a tractable means of controlling the resulting pullback approximation of $\pi_0$.

Finite-sample results for empirical KL minimization over triangular map classes give a complementary statistical perspective~\cite{irons2022triangular,wang2022minimax}. These results show that the error incurred in estimating a KR map from samples can be decomposed into two terms: a ``bias'' term that is essentially the approximation error described above and a ``variance'' term whose size depends on the size/complexity (specifically, the metric entropy in an appropriate norm) of the class of functions over which one estimates. While a more complex approximation class might reduce bias, it will increase the variance term for a fixed sample size. The implication for our work is that multifidelity formulations are expected to help insofar as they reduce the complexity of what must be learned from limited high-fidelity samples.

We propose two multifidelity families that pursue this complexity reduction in different ways. Section~\ref{sec:hierarchical_multifidelity} introduces hierarchical methods, which assume a fidelity ordering and compose maps between adjacent fidelity levels. Alternatively, Section~\ref{sec:nonhierarchical_multifidelity} introduces a non-hierarchical method using all lower-fidelity sources simultaneously through corrections to the high-fidelity map parameterization.

\subsection{Hierarchical multifidelity transport}\label{sec:hierarchical_multifidelity}

A common multifidelity strategy is to impose a hierarchy on the available information sources. Here, this means assuming that the target densities can be ordered as:
\begin{equation}\label{eqn:fidelity_hierarchy}
    \pi_M \rightarrow \pi_{M-1} \rightarrow \cdots \rightarrow \pi_0,
\end{equation}
where arrows indicate increasing fidelity. The premise is that %
we can estimate a complex high-fidelity KR map by composing low-complexity maps learned at successive fidelity levels. We consider two methods (see Figure~\ref{fig:hmf_changing_reference_and_fixed_reference}):%
\begin{itemize}
    \item The \emph{changing-reference} formulation finds a transport map between each adjacent pair of fidelities, where the lower-fidelity of the pair is the reference.
    \item The \emph{fixed-reference} formulation fixes $\eta$ at every stage
    but trains each new map on samples transformed by the composition of low-fidelity maps.%
\end{itemize}
Both formulations compose triangular maps to approximate $\SZero$. Their distinction is important. The first strategy mimics transports between adjacent fidelities; the second strategy preserves the reference's structure. Both approaches can be understood as %
extensions %
of deep transport methods~\cite{el2012bayesian,Marzouk2016,parno2018transport,peherstorfer2019transport}; see Appendix~\ref{app:hierarchical_transport_details}.%
\begin{figure}[!t]
    \centering
    \includegraphics[width=0.95\linewidth]{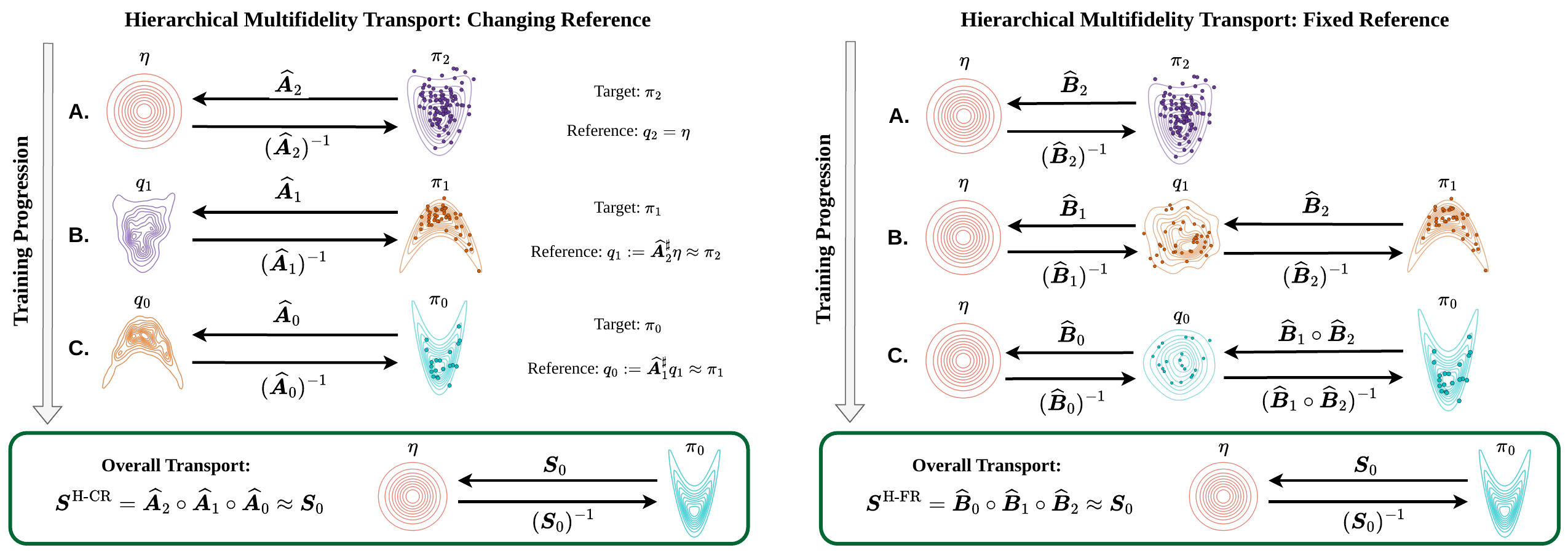}
    \caption{A generic tri-fidelity example of the changing reference hierarchical composition (left) and fixed reference hierarchical composition (right).%
    }
    \label{fig:hmf_changing_reference_and_fixed_reference}
\end{figure}

\subsubsection{Changing-reference hierarchical method}
\label{sec:hcr_formulation}

In the changing-reference formulation, the lowest-fidelity stage uses the prescribed reference density $\eta$; each subsequent stage uses the map-induced approximation of the previous lower-fidelity target density as its reference. The left panel of Figure~\ref{fig:hmf_changing_reference_and_fixed_reference} illustrates this construction.%

\paragraph{Formulation.}
The idealized compositional structure is

\begin{equation}\label{eqn:hcr_ideal_composition}
    \SZero
    =
    \bm{A}_M
    \circ
    \bm{A}_{M-1}
    \circ
    \cdots
    \circ
    \bm{A}_0 .
\end{equation}
Here, the map $\bm{A}_M$ pushes the lowest-fidelity density $\pi_M$ to $\eta$, while for $\ell=M-1,\ldots,0$, the map $\bm{A}_\ell$ pushes $\pi_\ell$ to an approximation of the density $\pi_{\ell+1}$ directly below it on the hierarchy.
In practice, we approximate each map $\mathbf{A}_\ell$ in sequence. Set%
\begin{equation}\label{eqn:hcr_initial_reference}
    q_M \coloneq \eta .
\end{equation}
Then, for $\ell=M,M-1,\ldots,0$, let $\widehat{\bm{A}}_\ell$ denote the learned triangular map satisfying $\widehat{\bm{A}}_\ell{}_{\sharp}\pi_\ell\approx q_\ell.$ %
In other words, we approximate the stage-$\ell$ target density,
\begin{equation}\label{eqn:hcr_stage_density}
    \widehat{\bm{A}}_\ell^{\sharp}q_\ell
    \approx
    \pi_\ell .
\end{equation}
For $\ell\geq 1$, the induced density becomes the reference for the next higher-fidelity:
\begin{equation}\label{eqn:hcr_reference_recursion}
    q_{\ell-1}
    \coloneq
    \widehat{\bm{A}}_{\ell-1}^{\sharp}q_\ell,
    \qquad
    \ell=M,\ldots,1 .
\end{equation}
Thus $q_{\ell-1}\approx\pi_\ell$ for $\ell=1,\ldots,M$, i.e., the reference at each stage is the learned approximation of the previous fidelity's density.
The final changing-reference map is:
\begin{equation}\label{eqn:hcr_final_map}
    \SCR
    \coloneq
    \widehat{\bm{A}}_M
    \circ
    \widehat{\bm{A}}_{M-1}
    \circ
    \cdots
    \circ
    \widehat{\bm{A}}_0
    \approx
    \SZero .
\end{equation}
If $\bm{x}\sim\pi_0$, then $\widehat{\bm{A}}_0(\bm{x})$ has density $q_0\approx\pi_1$ and $(\widehat{\bm{A}}_1\circ \widehat{\bm{A}}_0)(\bm{x})\sim q_1$ with $q_1\approx\pi_2$. Recursively, we get that $\SCR(\bm{x})$ is approximately distributed according to $\eta$.

\paragraph{Training.}
The maps are trained sequentially for $\ell=M,\ldots,0$. At stage $\ell$, reference $q_\ell$ is known via~\eqref{eqn:hcr_reference_recursion}. We then use the dataset $\mathcal X_\ell$ to find the map $\widehat{\bm{A}}_\ell$ minimizing:%
\begin{equation}\label{eqn:hcr_stage_objective}
    \min_{\bm{\theta}_\ell}
    \frac{1}{N_\ell}
    \sum_{j=1}^{N_\ell}
    \left[
        -\log q_\ell\!\left(
            \widehat{\bm{A}}_\ell(\bm{x}_\ell^{(j)};\bm{\theta}_\ell)
        \right)
        -
        \log\det\nabla \widehat{\bm{A}}_\ell
        (\bm{x}_\ell^{(j)};\bm{\theta}_\ell)
    \right]
    +
    \lambda^A_\ell
    \|\bm{\theta}_\ell\|_2^2,
\end{equation}
where $\lambda^A_{\ell}\geq 0$ is an $L^2$-regularization hyperparameter. For $\ell=M$, this is the usual KL objective with reference $q_M=\eta$. For $\ell<M$, the map-induced reference density $q_\ell$ can be evaluated recursively. The density of $q_{\ell-1}$ is defined by applying~\eqref{eqn:density_pullback} to the ansatz $q_{\ell-1} = \widehat{\bm{A}}_\ell^{\sharp}q_\ell$. %
This recursive density evaluation allows each learned approximation $\widehat{q}_\ell$ to become the reference for the next stage $\ell-1$.

Conceptually, the changing-reference construction transports between adjacent fidelity levels. Its main drawback is that, after the first stage, the reference densities $q_\ell$ are map-induced rather than prescribed; they are not generally log-concave or factorized. Thus, the later stage objectives need not retain the convexity or componentwise separability available for a Gaussian reference.

\subsubsection{Fixed-reference hierarchical composition}
\label{sec:hfr_formulation}

The fixed-reference formulation uses the same hierarchy \eqref{eqn:fidelity_hierarchy}, but trains every stage against the same prescribed reference density $\eta$. Rather than changing the reference, it changes the data: before training the stage-$\ell$ map, samples from $\pi_\ell$ are transformed by the composition of maps already learned from lower-fidelity levels. The right panel of Figure~\ref{fig:hmf_changing_reference_and_fixed_reference} illustrates this construction for a generic tri-fidelity hierarchy.

\paragraph{Formulation}
The idealized composition is given by:
\begin{equation}\label{eqn:hfr_ideal_composition}
    \SZero
    =
    \bm{B}_0
    \circ
    \bm{B}_1
    \circ
    \cdots
    \circ
    \bm{B}_M .
\end{equation}
The composition ordering differs from the changing-reference formulation because each higher-fidelity $\bm{B}_\ell$ is applied after the composition of previous low-fidelity maps.

At the lowest fidelity, $\ell=M$, no lower-fidelity maps have been learned, so the training data are simply those available for $\pi_M$:
\begin{equation}\label{eqn:hfr_lowest_transformation}
    \bm{y}_M^{(j)}
    :=
    \bm{x}_M^{(j)}
    \sim
    q_M
    \coloneq
    \pi_M,
    \qquad j=1,\ldots,N_M .
\end{equation}
We train $\widehat{\bm{B}}_M$ so that $(\widehat{\bm{B}}_M)_\sharp\pi_M\approx\eta$. At the next stage, $\ell=M-1$, samples from $\pi_{M-1}$ are first transformed by this learned lower-fidelity map:
\begin{equation}\label{eqn:hfr_M_minus_1_transformation}
    \bm{y}_{M-1}^{(j)}
    :=
    \widehat{\bm{B}}_M(\bm{x}_{M-1}^{(j)})
    \sim
    q_{M-1}
    \coloneq
    (\widehat{\bm{B}}_M)_\sharp\pi_{M-1}.
\end{equation}
If $\pi_{M-1}$ is close to $\pi_M$, then $(\widehat{\bm{B}}_M)_\sharp\pi_{M-1}\approx\eta$, leaving $\widehat{\bm{B}}_{M-1}$ to learn only a residual correction.

More generally, for $\ell=M-1,M-2,\ldots,0$, the transformed training samples are:%

\begin{equation}\label{eqn:hfr_transformation}
    \bm{y}_\ell^{(j)}
    :=
    \left(
    \widehat{\bm{B}}_{\ell+1}
    \circ
    \widehat{\bm{B}}_{\ell+2}
    \circ
    \cdots
    \circ
    \widehat{\bm{B}}_M
    \right)
    (\bm{x}_\ell^{(j)})
    \sim
    q_\ell
    \coloneq
    \left(
    \widehat{\bm{B}}_{\ell+1}
    \circ
    \widehat{\bm{B}}_{\ell+2}
    \circ
    \cdots
    \circ
    \widehat{\bm{B}}_M
    \right)_\sharp
    \pi_\ell .
\end{equation}
The stage map $\widehat{\bm{B}}_\ell$ is then trained so that
\begin{equation}\label{eqn:hfr_stage_push}
    (\widehat{\bm{B}}_\ell)_\sharp q_\ell
    \approx
    \eta .
\end{equation}
Thus, if the hierarchy is informative, the lower-fidelity composition approximately pushes $\pi_\ell$ to the reference, and the current stage learns only the remaining correction. The final fixed-reference hierarchical map is

\begin{equation}\label{eqn:hfr_final_map}
    \SFR
    \coloneq
    \widehat{\bm{B}}_0
    \circ
    \widehat{\bm{B}}_1
    \circ
    \cdots
    \circ
    \widehat{\bm{B}}_M
    \approx
    \SZero .
\end{equation}

\paragraph{Training}
Training proceeds sequentially for $\ell=M,M-1,\ldots,0$. At stage $\ell$, the transformed dataset $\{\bm{y}_\ell^{(j)}\}_{j=1}^{N_\ell}$ is defined by \eqref{eqn:hfr_lowest_transformation} for $\ell=M$ and by \eqref{eqn:hfr_transformation} otherwise. The map $\widehat{\bm{B}}_\ell$ is trained by solving
\begin{equation}\label{eqn:hfr_stage_objective}
    \min_{\bm{\theta}_\ell}
    \frac{1}{N_\ell}
    \sum_{j=1}^{N_\ell}
    \left[
        -\log \eta\!\left(
            \widehat{\bm{B}}_\ell(\bm{y}_\ell^{(j)};\bm{\theta}_\ell)
        \right)
        -
        \log\det\nabla \widehat{\bm{B}}_\ell
        (\bm{y}_\ell^{(j)};\bm{\theta}_\ell)
    \right]
    +
    \lambda_\ell^B
    \|\bm{\theta}_\ell\|_2^2 ,
\end{equation}
where $\lambda_\ell^B\geq 0$. For $\ell=M$, this is the ordinary KL objective; %
for $\ell<M$, it is the KL objective applied to samples transformed by the lower-fidelity composition.

The fixed-reference formulation has a computational advantage over the changing-reference formulation. If $\eta$ is standard Gaussian, every stage uses the same log-concave, factorized reference density and therefore retains the componentwise separability described in Section~\ref{sec:inverse_transport}. The key difference is that the intermediate maps are not transports between adjacent target densities; they are residual corrections to the lower-fidelity composition applied to the current target.

\subsection{Non-hierarchical multifidelity transport}
\label{sec:nonhierarchical_multifidelity}

The hierarchical constructions in Section~\ref{sec:hierarchical_multifidelity} require a prescribed fidelity ordering. In many multifidelity settings, such an ordering may be unavailable or unknown a priori; for example, different low-fidelity models may arise from distinct modeling assumptions, numerical discretizations, or reduced physics, so that no single hierarchy captures their relationship to the high-fidelity model. The non-hierarchical method developed here treats the low-fidelity densities as peer information sources. It uses all lower-fidelity datasets simultaneously and retains a common reference density $\eta$ for every fidelity level.

The central assumption is that the high-fidelity KR map $\SZero$ shares useful structure with at least some of the low-fidelity KR maps. Rather than composing maps along a hierarchy, the non-hierarchical method builds a single high-fidelity map approximation whose component parameterization contains corrected and scaled contributions from all low-fidelity maps. The trainable scale coefficients provide a mechanism for amplifying informative low-fidelity sources and suppressing uninformative ones, while the correction terms allow the low-fidelity maps to adapt during joint training.

\subsubsection{Formulation}
\label{sec:nonhierarchical_formulation}

For each low-fidelity density $\pi_\ell$, $\ell=1,\ldots,M$, let $\bm{S}_\ell$ denote the monotone triangular map satisfying
\begin{equation}\label{eqn:lf_kr_maps}
    (\bm{S}_\ell)_{\sharp}\pi_\ell=\eta,
    \qquad
    \ell=1,\ldots,M .
\end{equation}
We first approximate these low-fidelity maps independently. For $\ell=1,\ldots,M$ and $k=1,\ldots,d$, we write a map $\widehat{\bm{S}}_\ell$ componentwise for a set of parameters $\bm{\theta}_{\ell,k}$:
\begin{gather}\label{eqn:pre_trained_low_fidelity}
\widehat S_{\ell,k}(\bm{x}_{\leq k};\bm{\theta}_{\ell,k})
=
f_{\ell,k}(\bm{x}_{\leq k-1},0;\bm{\theta}_{\ell,k})
+
\int_{0}^{x_k}
g\!\left(
    \partial_k f_{\ell,k}(\bm{x}_{\leq k-1},t;\bm{\theta}_{\ell,k})
\right)\,\mathrm{d}t.
\end{gather}
Subsequently, we find the parameters $\bm{\theta}^*_{\ell,k}$ that minimize a regularized KL loss~\eqref{eqn:tri_map_kth_component_optimization}, and freeze them as \textit{pretrained} parameters, which encode a prior for the following correction step.

Because low-fidelity maps do not always help in approximating the high-fidelity transport, we introduce trainable corrections to each of the pretrained low-fidelity maps. We define the corrected low-fidelity map via:
\begin{gather}
    \widehat{S}^{\mathrm{corr}}_{\ell,k}(\bm{x}_{\leq k}; \bm{\theta}^{\mathrm{corr}}_{\ell,k}) = f_{\ell,k}(\bm{x}_{< k}, 0;\bm{\theta}^*_{\ell,k} + \bm{\theta}^{\mathrm{corr}}_{\ell,k}) + \int_0^{x_k}g\left(\partial_k f(x_{< k},t; \bm{\theta}^*_{\ell,k} + \bm{\theta}^{\mathrm{corr}}_{\ell,k})\right)\,\mathrm{d}t.
\end{gather}
As the correction is applied as a shift of the pretrained optimizer $\bm{\theta}^*$ for each fidelity's components, we maintain monotonicity of the function approximation. We remark, though, that $\bm{\theta}^*_{\ell,k}$ is entirely fixed after pretraining.
We now define the non-hierarchical multifidelity map $\SNH$, whose $k$-th component combines a high-fidelity shift function $\delta_k$ with scaled contributions from the corrected low-fidelity components:

\begin{equation}\label{eqn:nonhier_parameterization}
\begin{gathered}
    S^{\mathrm{NH}}_k(\bm{x}_{\leq k};\bm{\Theta}_k) = \delta_{k}(\bm{x}_{<k},0;\bm{\vartheta}_k) + \sum_{\ell=1}^M \rho_{\ell,k}^{\mathrm{nm}}f_{\ell,k}(\bm{x}_{<k},0;\bm{\theta}^*_{\ell,k} + \bm{\theta}^{\mathrm{corr}}_{\ell,k})\\[-0.5em]
    + \int_0^{x_k}g\Big(\partial_k \delta_k(\bm{x}_{< k},t;\bm{\vartheta}_k) + \sum_{\ell=1}^M \rho^{\mathrm{m}}_{\ell,k} \partial_k f_{\ell,k}(\bm{x}_{<k},t; \bm{\theta}^*_{\ell,k} + \bm{\theta}^{\mathrm{corr}}_{\ell,k})\Big)\,\mathrm{d}t,\quad k=1,\ldots,d.
\end{gathered}
\end{equation}
Here, we parameterize $\SNH$ via \(\bm{\Theta}_k
=
\left(
    \bm{\vartheta}_k,
    \{\rho^{\mathrm{nm}}_{\ell,k},\rho^{\mathrm{m}}_{\ell,k},
    \bm{\theta}^{\mathrm{corr}}_{\ell,k}\}_{\ell=1}^{M}
\right)\), where coefficient $\rho^{\mathrm{nm}}_{\ell,k}$ %
scales the nonmonotone contribution of fidelity $\ell$ to component $k$, while $\rho^{\mathrm{m}}_{\ell,k}$ scales the corresponding monotone contribution. The high-fidelity shift $\delta_k$ is a correction learned only from high-fidelity data. Placing the monotone linear combination inside the rectifier ensures a monotone $S^{\mathrm{NH}}_k$ without constraining the trainable parameters.

This parameterization is a monotonicity-preserving scale-shift construction. It is analogous to scale-shift models in the multifidelity surrogate literature \cite{kennedy2000predicting,GorodetskyCM,GorodetskyIJUQ}, but the scaling and shifting are applied separately to the nonmonotone and monotone parts of the triangular-map parameterization. The scale coefficients can downweight low-fidelity sources that are not useful for the high-fidelity map, while the correction functions allow the shared low-fidelity structure to adapt during joint training. In the present work, the scale coefficients are trainable scalars and the high-fidelity shift function lies in a subset of the space of functions for $f_{\ell,k}$. A straightforward extension would involve making the scale coefficients non-constant functions as well, and then balancing expressivity between the high-fidelity shift and scale functions.

\subsubsection{Training}
\label{sec:nonhierarchical_training}

Training proceeds in two phases: First, the low-fidelity maps are pretrained independently. Second, the corrected low-fidelity maps and the high-fidelity non-hierarchical map are trained jointly.

\paragraph{Phase 1: independent low-fidelity pretraining}
For each low-fidelity level $\ell=1,\ldots,M$, we pretrain $\widehat{\bm{S}}_\ell$ by minimizing the empirical KL loss
\begin{equation}\label{eqn:tri_map_optimization_pre_train_lf}
    \mathcal{J}^{\mathrm{pre}}_{\ell}(\bm{\theta}_{\ell})
    :=
    \frac{1}{N_\ell}
    \sum_{j=1}^{N_\ell}
    \left[
        -\log\eta\!\left(
            \widehat{\bm{S}}_\ell(\bm{x}^{(j)}_\ell;\bm{\theta}_{\ell})
        \right)
        -
        \log\det\nabla\widehat{\bm{S}}_\ell
        (\bm{x}^{(j)}_\ell;\bm{\theta}_{\ell})
    \right].
\end{equation}
These $M$ problems decouple across fidelity levels and can be solved independently. If $\eta=\mathcal{N}(0,I)$, each pretraining objective also separates across triangular components:

\begin{equation}\label{eqn:tri_map_kth_component_optimization_pre_train_lf}
    \mathcal{J}^{\mathrm{pre}}_{\ell,k}(\bm{\theta}_{\ell,k})
    :=
    \frac{1}{N_\ell}
    \sum_{j=1}^{N_\ell}
    \left[
        \frac{1}{2}
        \widehat S_{\ell,k}^2(\bm{x}^{(j)}_\ell;\bm{\theta}_{\ell,k})
        -
        \log \partial_k\widehat S_{\ell,k}(\bm{x}^{(j)}_\ell;\bm{\theta}_{\ell,k})
    \right],
    \qquad
    k=1,\ldots,d .
\end{equation}

\paragraph{Phase 2: all-at-once multifidelity training}
After low-fidelity pretraining, the pretrained functions $f_{\ell,k}$ are fixed. We then jointly train the high-fidelity shift functions $\delta_k$, the scale coefficients $\rho^{\mathrm{nm}}_{\ell,k}$ and $\rho^{\mathrm{m}}_{\ell,k}$, and the correction functions $f^{\mathrm{corr}}_{\ell,k}$ by minimizing

\begin{equation}\label{eqn:NH-MF_all_at_once}
\begin{gathered}
\mathcal{J}(\bm{\Theta})
:=
\frac{1}{N_0}
\sum_{j=1}^{N_0}
\left[
    -\log\eta\!\left(
        \SNH(\bm{x}^{(j)}_0;\bm{\Theta})
    \right)
    -
    \log\det\nabla \SNH
    (\bm{x}^{(j)}_0;\bm{\Theta})
\right]
\\[-0.5em]
+
\sum_{\ell=1}^{M}
\frac{w_\ell}{N_\ell}
\sum_{j=1}^{N_\ell}
\left[
    -\log\eta\!\left(
        \widehat{\bm{S}}^{\mathrm{corr}}_\ell
        (\bm{x}^{(j)}_\ell;\bm{\theta}^{\mathrm{corr}}_\ell)
    \right)
    -
    \log\det\nabla
    \widehat{\bm{S}}^{\mathrm{corr}}_\ell
    (\bm{x}^{(j)}_\ell;\bm{\theta}^{\mathrm{corr}}_\ell)
\right]
+
\lambda\|\bm{\Theta}\|_2^2,
\end{gathered}
\end{equation}\ignorespacesafterend
where $\lambda\geq0$ is a hyperparameter controlling the strength of the $L^2$ regularization.
For $\ell=1,\ldots,M$, we define each low-fidelity data-dependent weight as $w_\ell = N_\ell / \sum_{r=1}^M N_r$. %
The high-fidelity term is left with unit weight because the primary objective is to approximate a map pushing $\pi_0$ to $\eta$. Other weighting strategies are possible, and we leave the exploration of additional choiches to future work.%

The first term in~\eqref{eqn:NH-MF_all_at_once} trains $\SNH$ to push high-fidelity samples to the reference. The second term trains each corrected low-fidelity map to push its own low-fidelity samples to the same reference. Since $\SNH$ shares correction parameters with the corrected low-fidelity maps, the second term acts as a multifidelity regularizer: it biases the high-fidelity map toward structures supported by abundant low-fidelity data. At the same time, the scale coefficients $\rho^{\mathrm{nm}}_{\ell,k}$ and $\rho^{\mathrm{m}}_{\ell,k}$ allow the high-fidelity map to suppress low-fidelity contributions that are not useful for approximating $\SZero$.

Because every term in \eqref{eqn:NH-MF_all_at_once} uses the same reference density $\eta$, this non-hierarchical multifidelity (NH-MF) objective retains the fixed-reference structure discussed in Section~\ref{sec:conditions_of_effectiveness}. For a factorized reference, this structure also yields componentwise separability; a general statement covering both the non-hierarchical objective and the fixed-reference hierarchical stages is given in Appendix~\ref{app:fixed_reference_separability}.

\subsection{Discussion and comparison}
\label{sec:conditions_of_effectiveness}
The common objective of all three methods is to reduce the complexity of what must be learned from scarce high-fidelity samples, but the methods impose different structural assumptions. The changing-reference hierarchy assumes that the difficult transport problem between $\pi_0$ to $\eta$ can be decomposed into simpler transport problems along a fidelity ordering. On the other hand, the fixed-reference hierarchy assumes that maps learned at lower fidelities {also} approximately map higher-fidelity samples to the reference, so that later stages need only to learn low-complexity residual corrections. %
The non-hierarchical method assumes that the high-fidelity KR map is represented well by a scale-shift construction built from corrected low-fidelity maps.

The structural assumption in the hierarchical approaches is that each transport map in the compositions \eqref{eqn:hcr_ideal_composition} or \eqref{eqn:hfr_ideal_composition}, especially for smaller values of $\ell$,
is simpler than the direct map from $\pi_0$ to $\eta$. This is natural when fidelities arise from a nested model family, such as increasing mesh resolution or decreasing solver tolerance. Under this condition, each map may belong to a less complex function class than the full high-fidelity KR map, which is precisely the type of complexity reduction suggested by the approximation and learning theory discussed in the opening of Section~\ref{sec:multifidelity_inverse_transport}. 
The vulnerability is that successive maps would need to become simpler, which is not guaranteed to be true in general. For instance, in the changing-reference formulation a learned reference could strongly depart from the target making the next learning problem harder. %
Furthermore, as previously noted, after the lowest-fidelity stage, the reference densities are learned map-induced densities rather than the prescribed reference $\eta$, so the later stage objectives generally do not retain the convexity or componentwise separability available when the reference is standard Gaussian. %

On the other hand, the fixed-reference method is effective when any given lower-fidelity transport map resembles the true KR map $\SZero$, %
so each residual map $\widehat{\bm{B}}_\ell$ is nearly the identity. %
The fixed-reference formulation seeks complexity reduction through residual learning rather than through explicit adjacent transports. Its advantage is that every stage is trained against the same prescribed reference; when $\eta$ is Gaussian, this preserves the componentwise separability described in Appendix~\ref{app:fixed_reference_separability}. %
The main limitation of fixing the reference is that we still require a favorable fidelity ordering, and errors in the cumulative lower-fidelity map can still impact all higher-fidelity residual stages when working with finite samples. For instance, target samples from lower-fidelity models could be mapped far from the support of the standard Gaussian reference making the use of this formulation challenging.

The non-hierarchical method does not assume a fidelity ordering. Its structural condition is representational: the high-fidelity map $\SZero$ should be well approximated by the non-hierarchical parameterization, namely by a high-fidelity shift combined with scaled and corrected low-fidelity map components. This construction can reduce statistical complexity when some low-fidelity maps contain structure aligned with the high-fidelity transport, because the high-fidelity samples are then used primarily to learn scale coefficients, corrections, and residual high-fidelity structure rather than the full map from scratch. The trainable scales provide a mechanism to suppress uninformative sources, while the correction terms allow useful low-fidelity structure to adapt. This makes non-hierarchical most appropriate for heterogeneous model ensembles with no reliable ordering; such settings are also those in which peer-structured multifidelity surrogate methods have been observed to be more robust than purely hierarchical ones~\cite{GorodetskyCM,GorodetskyIJUQ}. The method may fail when the low-fidelity transports do not contain structure useful for approximating $\SZero$, or when it is more difficult to learn the low-fidelity correction than to learn the high-fidelity map directly.%

Finally, we discuss how extensible these methods are to other related approaches in transport literature. Though we use triangular maps in our fixed- and changing-reference hierarchical methods, these multifidelity methods actually impose no requirement on the form of the map. Instead, we simply require that the map $\widehat{\bm{S}}$ allow easy evaluation of the pullback density $\widehat{\bm{S}}^\sharp\eta$. In fact, we do not even require use of the KL objective. Both these options are chosen, however, for their computational tractability in the data-sparse regime that is common in high-fidelity modeling. By contrast, the non-hierarchical method is more directly tied to the triangular transport formulation. Its construction uses the componentwise integrated-rectifier parameterization to combine high-fidelity shifts with scaled and corrected low-fidelity map components while preserving monotonicity. This mechanism does not immediately translate to transport formulations without an analogous componentwise monotonicity-preserving structure. Nevertheless, the broader idea of using low-fidelity transport maps as peer information sources, rather than arranging them in a hierarchy, is not limited to triangular maps. Developing non-hierarchical multifidelity constructions for other classes of transport is a natural and compelling direction for future work.

\section{Numerical results}\label{sec:numerics}
We evaluate the proposed multifidelity transport constructions on three numerical examples. We first examine two map estimation tasks in which the high-fidelity KR map $\SZero$ is known explicitly: a two-dimensional polynomial example in Section~\ref{sec:banana_example} and a five-dimensional softened Neal's funnel example in Section~\ref{sec:nealfunnel_example}. The third example in Section~\ref{sec:LV} examines %
simulation-based inference in a stochastic Lotka--Volterra model, where the exact KR map is unavailable. Across these examples, we explore whether lower-fidelity data can reduce reliance on high-fidelity training samples in constructing useful transport maps.
All experiments compare single-fidelity triangular transport, trained only from high-fidelity samples, with the %
multifidelity methods developed in Sections~\ref{sec:hierarchical_multifidelity} and~\ref{sec:nonhierarchical_multifidelity}.

For the first two examples, sample allocations and hierarchical orderings are organized through the cost models in Table~\ref{tab:numerical_cost_models}. For each high-fidelity sample size $N_0$, the low-fidelity sample sizes are denoted by $N_1$ and $N_2$, and the hierarchy column gives the ordering used by the hierarchical methods, with arrows indicating increasing fidelity. Cost models A and B exchange the sample allocation between $\pi_1$ and $\pi_2$, while cost models C and D reverse the hierarchy with equal low-fidelity sample sizes. The Lotka--Volterra example uses a reduced cost-model setup described in Section~\ref{sec:LV}.

\begin{table}[!htb]
\centering
\caption{Cost models used across numerical experiments. For each high-fidelity sample size $N_0$, the table lists the low-fidelity sample sizes and the hierarchy used by the hierarchical multifidelity methods. }
\label{tab:numerical_cost_models}

\begin{tabular}{@{}c c c c@{}}
\toprule
Cost Model (CM) & Fidelity Hierarchy & $N_1/N_0$ & $N_2/N_0$ \\ \midrule
A & $\pi_2 \to \pi_1 \to \pi_0$ & $10$ & $100$ \\
B & $\pi_1 \to \pi_2 \to \pi_0$ & $100$ & $10$ \\
C & $\pi_2 \to \pi_1 \to \pi_0$ & $100$ & $100$ \\
D & $\pi_1 \to \pi_2 \to \pi_0$ & $100$ & $100$ \\
\bottomrule
\end{tabular}
\end{table}

All maps use the integrated-rectifier parameterization and Hermite-function expansions described in Appendix~\ref{app:map_param}. We use \emph{total order} (TO) to denote the maximum total polynomial order of these expansions. Thus, a TO-$p$ map uses total order $p$ for each relevant map component or stage. We also use \emph{high-fidelity total order} (HF-TO) to isolate the complexity of the part of the construction trained from high-fidelity samples: the final high-fidelity stage for hierarchical maps, and the high-fidelity shift and correction functions in \eqref{eqn:nonhier_parameterization} for the non-hierarchical map.

Samples are standardized componentwise at each fidelity level using the empirical mean and standard deviation of the training data. All training objectives include an $L^2$ regularization penalty, with parameters selected by cross validation as described in Appendix~\ref{app:regularization_cv}. For hierarchical maps, regularization is applied only to the final stage trained from high-fidelity data. For non-hierarchical maps, regularization is applied to the high-fidelity shift function, while scale coefficients, pretrained low-fidelity maps, and low-fidelity correction functions are left unregularized.

The accuracy diagnostics depend on the example. For the first two examples, where $\SZero$ is known, we report the relative mean-squared error in the map,

\begin{equation}\label{eqn:target_measure_weighted_error}
    \mathcal{E}_{\mathrm{rel}}(\widehat{\bm{S}}) :=
    \mathbb{E}_{\pi_0}\left[\left\|\widehat{\bm{S}}(X)-\SZero(X)\right\|_2^2 \right]\,\bigg/\,\mathbb{E}_{\pi_0}\left[\left\|\SZero(X)\right\|_2^2\right].
\end{equation}

\subsection{Simple analytical verification example}
\label{sec:banana_example}
We begin with a $d=2$ dimensional example for which the high-fidelity KR map is known explicitly. This allows us to evaluate map error directly and compare single-fidelity, hierarchical multifidelity, and non-hierarchical multifidelity constructions under controlled conditions. The example is designed to explore three concerns: performance under a common na\"ive parameterization, whether multifidelity structure can reduce the complexity of the part of the map learned from scarce high-fidelity data, and how sensitive the methods are to low-fidelity sample allocation or the chosen fidelity hierarchy.

Consider three banana-shaped target densities, with high-fidelity density $\pi_0$ and low-fidelity densities $\pi_1,\pi_2$. For $\ell\in{0,1,2}$, define
\begin{equation}\label{eqn:banana_density_family}
\pi_{\ell}(x_1,x_2)
\propto
\mathcal{N}_{1}\!\left(\frac{x_1-d_\ell}{a_\ell}\right)
\mathcal{N}_{1}\!\left(
\frac{x_2-e_\ell-h_\ell((x_1-d_\ell)/a_\ell)}{b_\ell}
\right),\quad h_\ell(z) = \sum_{k=2}^4 c_{\ell,k}z^k,
\end{equation}
where $\mathcal{N}_1$ is the one-dimensional standard Gaussian density. %
\begin{figure}[!tb]
\centering
\includegraphics[width=0.85\linewidth]{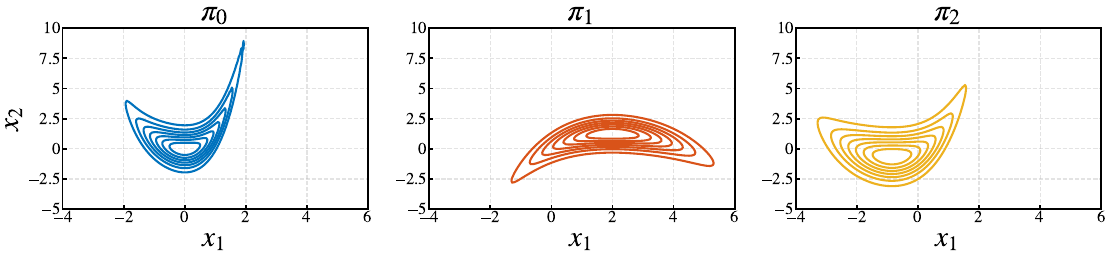}
\caption{Section~\ref{sec:banana_example} target densities: high-fidelity $\pi_0$ and low-fidelity $\pi_1,\pi_2$.}
\label{fig:tri-fid_target_triptych}
\end{figure}
Figure~\ref{fig:tri-fid_target_triptych} shows the three target densities.
The density $\pi_2$ is more aligned with $\pi_0$ than $\pi_1$, making this example useful for testing sensitivity to both low-fidelity sample allocation and hierarchical ordering. The exact high-fidelity KR map lies in a TO-4 triangular map class. Reproducibility details for this example, including the specific fidelity parameters and expression for the high-fidelity KR map, are collected in Appendix~\ref{app:banana_example}.

We compare four approximations: the single-fidelity map $\SSF$ trained using only high-fidelity samples; the non-hierarchical multifidelity map $\SNH$ from Section~\ref{sec:nonhierarchical_multifidelity}; the fixed-reference hierarchical map $\SFR$ from Section~\ref{sec:hfr_formulation}; and the changing-reference hierarchical map $\SCR$ from Section~\ref{sec:hcr_formulation}. Accuracy is measured using a sample approximation %
of \eqref{eqn:target_measure_weighted_error} with $10^4$ test samples and with $\SZero$ given by \eqref{eqn:banana4_true_kr}. For each method and each $N_0\in\{20, 100, 180, 260\}$, we repeat training over 20 independent random training data realizations.
\begin{figure}[!htb]
    \centering
    \includegraphics[width=0.9\linewidth]{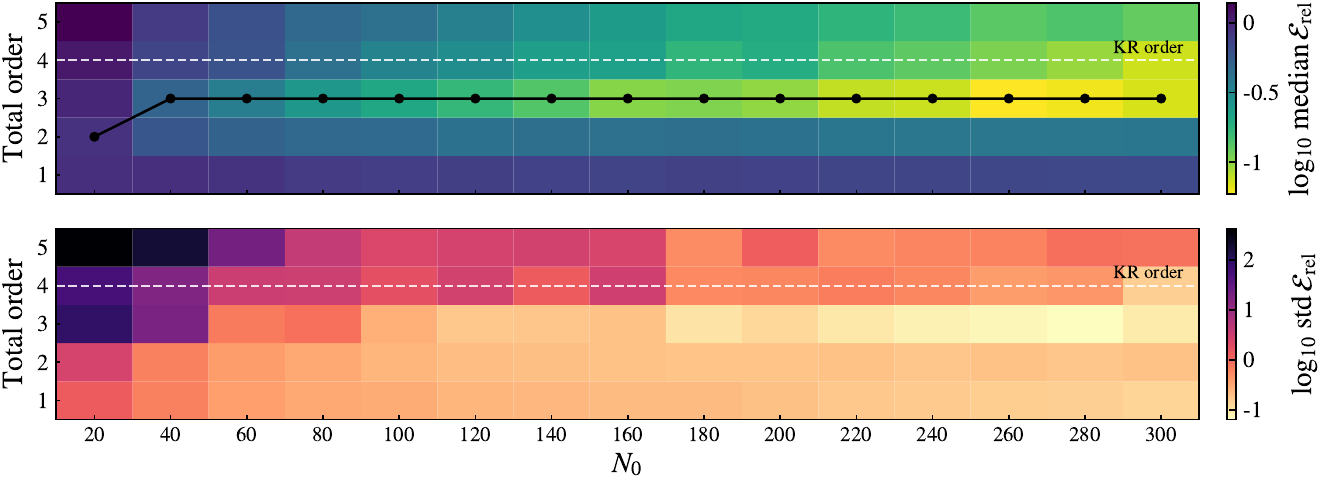}
    \caption{Section~\ref{sec:banana_example} single-fidelity transport performance as a function of total order and $N_0$. Top: median relative map error over 100 training replications on 100 random realizations of training data for each $N_0$; the black line shows the total order with the best median error for each $N_0$. Bottom: the standard deviation in the relative error over the same ensemble.}
    \label{fig:banana_heatmap}
\end{figure}

Figure~\ref{fig:banana_heatmap} shows the performance of $\SSF$ across a wide range of sparse $N_0$ values and total orders from one to five.
These results illustrate the opportunity for multifidelity improvement in sparse data regimes. Although the true KR map lies in a TO-4 approximation class, a TO-3 parameterization yields both lower median error and reduced variability when $N_0$ is small. This suggests that the order 4 terms in the true KR map are not reliably identifiable from sparse high-fidelity data alone. Consequently, if these higher-order features can instead be approximated through low-complexity transformations of well-resolved lower-fidelity features, then a multifidelity construction can provide an advantage in this regime.
 
\begin{figure}[!htb]
    \centering
    \includegraphics[width=\linewidth]{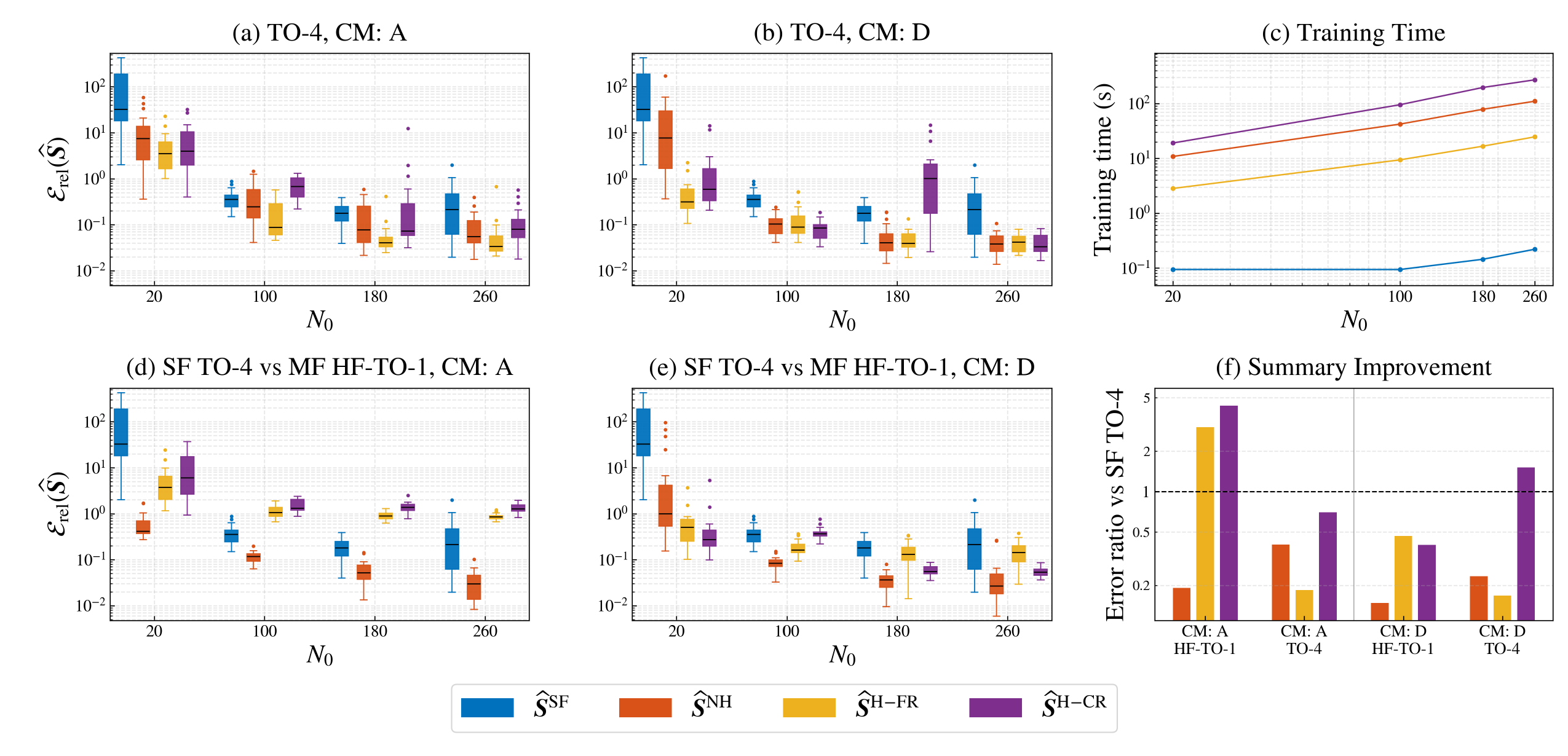}
    \caption{Quantitative comparisons of transport maps in Section~\ref{sec:banana_example}. Panels (a,b) compare TO-4 maps for cost models A and D. Panel (c) shows median training time for TO-4 maps time versus $N_0$, averaged over all cost models. Panels (d,e) compare the single-fidelity TO-4 baseline with multifidelity maps using high-fidelity TO-1 under cost models A and D. Panel (f) summarizes median error ratios relative to the single-fidelity TO-4 baseline, averaged over all $N_0$; values below one indicate improvement over the baseline.}
    \label{fig:quartic_banana_overall_results}
\end{figure}

Figure~\ref{fig:quartic_banana_overall_results} provides a comparison between the methods with respect to map parameterization, cost model, high-fidelity training data allotment $N_0$, and training time. Panels (a,b) use a common TO-4 parameterization. This is the lowest total order class that contains the high-fidelity KR map \eqref{eqn:banana4_true_kr}, so the single-fidelity parameterization is well specified. For $\SNH$, however, this parameterization is deliberately permissive: the high-fidelity shift can represent the exact map without relying on the pretrained low-fidelity maps. Thus, $\SNH$ is not forced to use low-fidelity structure---its use is \emph{implicit}. The optimizer is free to ignore the pretrained low-fidelity maps and represent the high-fidelity map primarily through the shift function.

The hierarchical methods use low-fidelity information differently. Their multifidelity structure is imposed \emph{explicitly} as a sequence of transport problems: the fixed-reference formulation transforms samples by the lower-fidelity map composition before training a final stage, while the changing-reference formulation changes the reference density along the hierarchy. Thus, even with a permissive HF-TO-4 parameterization, the hierarchical models explicitly modify the high-fidelity learning problem using lower-fidelity maps. Panels (a,b) should therefore be interpreted as a baseline comparison of the methods under a shared generic parameterization that is well-specified for a single fidelity, not as the parameterization most tailored to effectively leveraging multifidelity information and reducing sample complexity as outlined in Section~\ref{sec:conditions_of_effectiveness}.

Under this generic parameterization, the multifidelity methods generally improve upon single-fidelity transport, except in the sparsest high-fidelity data regimes where all methods produce relatively inaccurate maps. As $N_0$ increases, the performance gaps shrink, consistent with the high-fidelity data becoming sufficient to estimate the true KR map directly. Across the displayed regimes, $\SFR$ gives the most consistent improvement, followed by $\SNH$. The changing-reference hierarchical method is also often beneficial, but exhibits greater variability at some sample sizes. This behavior is consistent with the discussion in Section~\ref{sec:hcr_formulation}: changing-reference training uses map-induced intermediate references, resulting in more challenging, higher-dimensional, and non-convex map optimization. %
The multifidelity methods often perform better under cost model D, which provides more low-fidelity samples overall than cost model A, and for the hierarchical methods, uses the better hierarchy---placing the more informative low-fidelity density $\pi_2$ directly below $\pi_0$.

The timing results in panel (c) are included primarily to compare the two hierarchical approaches. In multifidelity settings, the dominant cost is typically assumed to be data acquisition rather than map training, so the faster training time of $\SSF$ is not the main comparison of interest. The key observation, however, is that $\SCR$ is substantially more expensive than $\SFR$, consistent with the recursive map-induced reference evaluation in Section~\ref{sec:hcr_formulation} and the loss of componentwise separability after the first stage. Since the changing-reference hierarchy does not provide a clear accuracy advantage commensurate with this additional cost, the remaining numerical examples use the fixed-reference formulation %
as hierarchical multifidelity method.

Panels (d,e) test the intended sample-complexity mechanism more directly. The single-fidelity baseline retains a TO-4 parameterization, while the multifidelity methods use HF-TO-1. Thus, the parts of the map learned from scarce high-fidelity samples are restricted, while the low-fidelity maps in NH and the pre-final-stage maps in the H-CR and H-FR methods retain TO-4. Under this comparison, $\SNH$ achieves the clearest improvement over the TO-4 single-fidelity baseline across the displayed sample sizes. This indicates that, in this example, the pre-trained low-fidelity components and trainable scales can represent much of the structure that would otherwise need to be learned from high-fidelity samples alone. The hierarchical methods again depend strongly on the ordering. Under cost model D, where $\pi_2$ is placed directly below $\pi_0$, the hierarchical maps are competitive in several regimes. Under cost model A, the same methods are much less effective.

Panel (f) aggregates these trends by reporting median error ratios relative to the
single-fidelity TO-4 baseline averaged over all $N_0$. The best aggregate performance is obtained by NH with HF-TO-1 under cost model D. H-FR with the favorable hierarchy and TO-4 parameterization is similarly effective, but the hierarchical methods are more sensitive to the chosen ordering, especially under the HF-TO-1 parameterization. 

\begin{figure}[!htb]
\centering
\includegraphics[clip, trim={0 0.1cm 0 0}, width=0.95\linewidth]{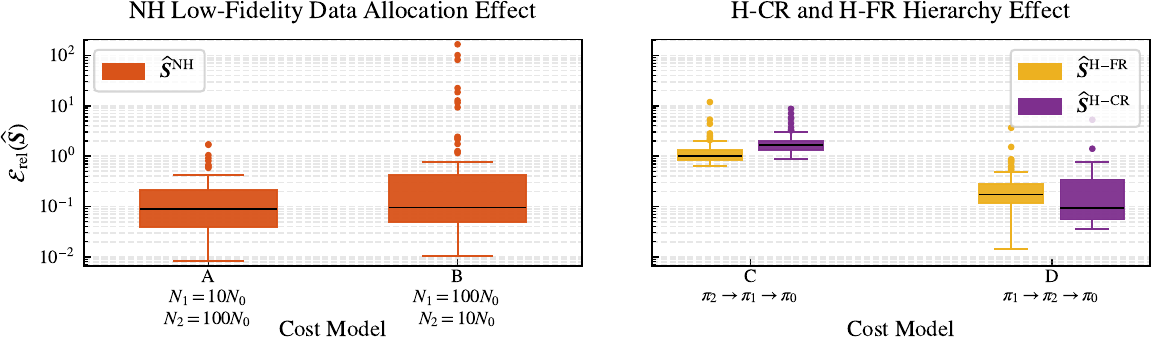}
\caption{Sensitivity to low-fidelity sample allocation and hierarchical ordering for HF-TO-1 multifidelity maps in Section~\ref{sec:banana_example} aggregated over $N_0\in\{20,100,180,260\}$. Left: $\SNH$ under cost models A and B, compared under different budgets for $\pi_1$ and $\pi_2$. %
Right: $\SFR$ and $\SCR$ under cost models C and D, which use equal low-fidelity sample sizes but reverse the hierarchy order.
}
\label{fig:banana_hierarchy_and_allocation_sensitivity}
\end{figure}

Figure~\ref{fig:banana_hierarchy_and_allocation_sensitivity} isolates the low-fidelity allocation and hierarchy effects described in Table~\ref{tab:numerical_cost_models}. The left panel compares $\SNH$ under cost models A and B. Since NH is insensitive to the fidelity hierarchy, this comparison isolates the effect of exchanging the low-fidelity sample allocation between $\pi_1$ and $\pi_2$. The difference is modest, but allocating more low-fidelity samples to the more informative source $\pi_2$ improves the aggregate error. This is consistent with the weighted low-fidelity terms in the NH objective \eqref{eqn:NH-MF_all_at_once}: the low-fidelity source with more samples receives greater weight in the joint training objective. The effect is modest, indicating that NH is not highly sensitive to the precise low-fidelity sample allocation in this example. 

The right panel compares the HF-TO-1 hierarchical methods under cost models C and D. These cost models use the same low-fidelity sample sizes but reverse the hierarchy order, thereby isolating the effect of hierarchy choice. The hierarchy used in cost model D performs substantially better than the hierarchy used in cost model C. This is consistent with the geometry in Figure~\ref{fig:tri-fid_target_triptych}: placing the more informative low-fidelity density $\pi_2$ directly below $\pi_0$ reduces the complexity of the final high-fidelity stage. This behavior reflects the conditions discussed in Section~\ref{sec:conditions_of_effectiveness}, where hierarchical methods are effective when adjacent transports, or the fixed-reference's residuals that are pushed forward, are simple enough to estimate from the available samples. A poor ordering can produce substantially larger errors.

Overall, this example supports three conclusions. First, multifidelity information
can improve triangular map estimation when high-fidelity data are limited. Second, the way multifidelity information enters the construction matters: hierarchical methods use low-fidelity maps explicitly through staged transformations, and under a suitable hierarchy, this empirically provides resiliency to a na\"\i ve final-stage parameterization. By contrast, NH uses low-fidelity information through its parameterization and benefits from a parameterization that biases it towards leveraging low-fidelity structure; e.g., by restricting the expressivity of the shift function. Third, NH is insensitive to the fidelity hierarchy, whereas H-FR and H-CR can be highly effective under a favorable ordering but degrade under a poor one. Among the hierarchical variants, H-CR is substantially more expensive to train and does not provide a commensurate accuracy advantage in this example.

\subsection{Softened Neal's funnels}
\label{sec:nealfunnel_example}

We next consider a $d=5$ dimensional tri-fidelity example \cite{neal2003slice} with an explicit high-fidelity KR map. This example is designed as a higher-dimensional, non-polynomial stress test of the conclusions from Section~\ref{sec:banana_example}. Unlike the previous example, the high-fidelity map is not contained in a finite polynomial total order class. The low-fidelity ensemble again contains one structurally aligned source and one less aligned source.

Let $d=5$ and $z=(z_1,\ldots,z_d)\sim\mathcal N(0,I_d)$. For $\ell\in{0,1,2}$, define $x\sim\pi_\ell$ by
\begin{equation}\label{eqn:nealfunnel_forward}
x_1=\mu_\ell+\sigma_\ell z_1,\qquad
x_k=\tau_{\ell,k}(z_1)\,z_k+\xi_{\ell,k}(z_1),\quad k=2,\ldots,d.
\end{equation}
Here $\tau_{\ell,k}$ and $\xi_{\ell,k}$ are fidelity- and coordinate-dependent conditional scale and shift functions. Their precise definitions, the coordinate-specific modifiers, and the fidelity parameters are given in \eqref{eqn:nealfunnel_scale_shift}--\eqref{eqn:nealfunnel_params}. These choices produce softened funnel-shaped targets whose conditional scales vary nonlinearly with $x_1$. The high-fidelity density is denoted by $\pi_0$, while $\pi_1$ and $\pi_2$ are lower-fidelity densities chosen so that $\pi_2$ is structurally better aligned with $\pi_0$ than with $\pi_1$.

\begin{figure}[!htb]
\centering
\includegraphics[width=\linewidth]{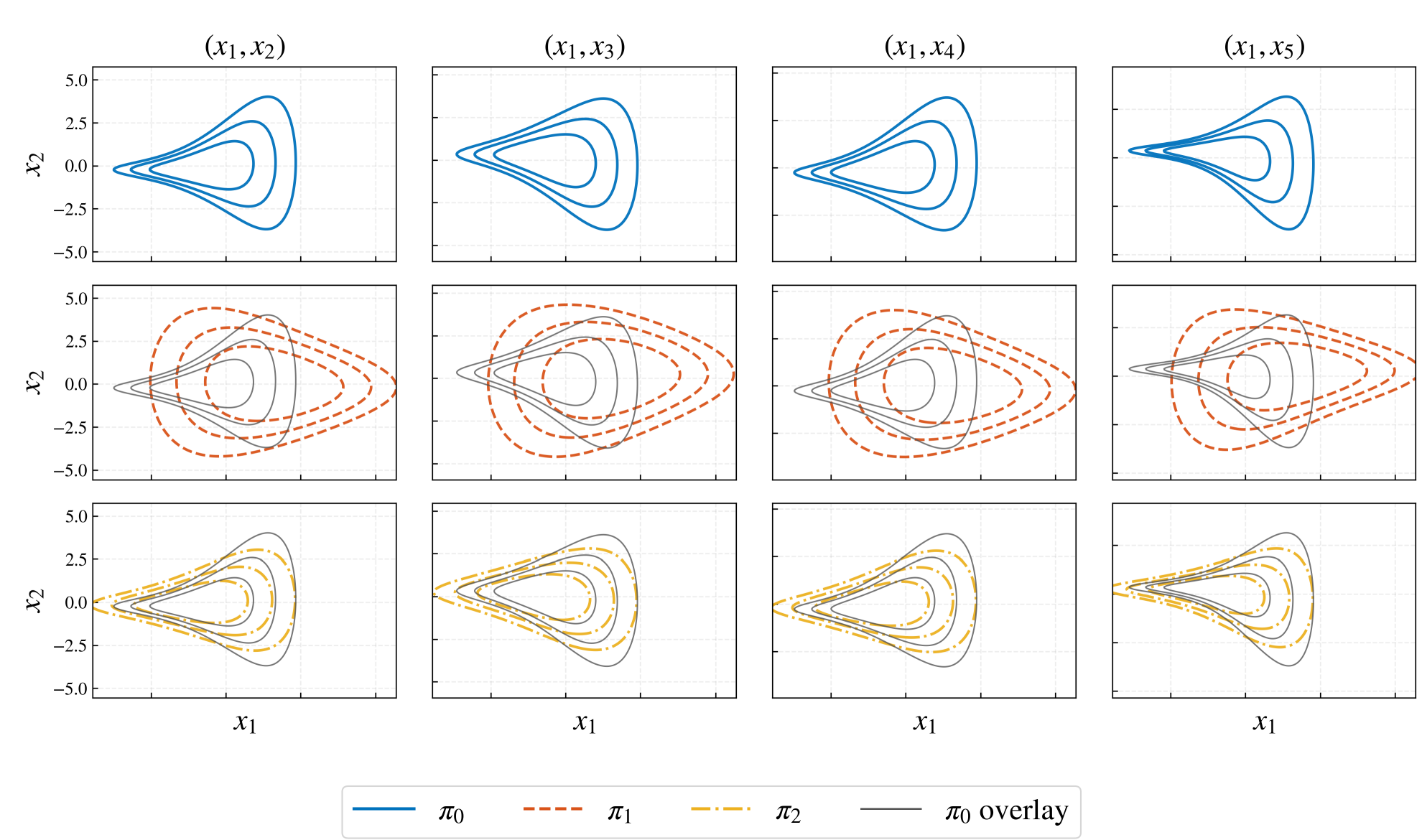}
\caption{Contour plots of $(x_1,x_k)$ for Section~\ref{sec:nealfunnel_example}. %
Colors denote high-fidelity $\pi_0$ and low-fidelity $\pi_1,\ \pi_2$. In the low-fidelity rows, contours of $\pi_0$ are overlaid for reference. The figure highlights that $\pi_2$ is structurally better aligned with $\pi_0$, while $\pi_1$ exhibits a more pronounced mismatch of geometry.}
\label{fig:nealfunnel_vis}
\end{figure}
Figure~\ref{fig:nealfunnel_vis} %
shows the three target densities through their $(x_1,x_k)$ marginals. The density $\pi_2$ preserves the orientation and broad conditional-scale structure of $\pi_0$, while $\pi_1$ is intentionally misaligned by reversing the funnel orientation. Thus the geometry again suggests that cost models B and D, which place $\pi_2$ directly below $\pi_0$, should be more favorable for the hierarchical method than cost models A and C.

For $\pi_0$, the exact KR map is explicit and is reported in \eqref{eqn:nealfunnel_true_kr}. We compare $\SSF$, $\SNH$, and $\SFR$, using a sample approximation of \eqref{eqn:target_measure_weighted_error} with $10^4$ samples (we have excluded the hierarchical changing-reference method following its poor numerical properties previously discussed). %
For each method and each $N_0\in\{50, 150, 250, 350, 800\}$, we repeat training over 20 random training data realizations.
\begin{figure}[!htb]
    \centering
    \includegraphics[width=0.95\linewidth]{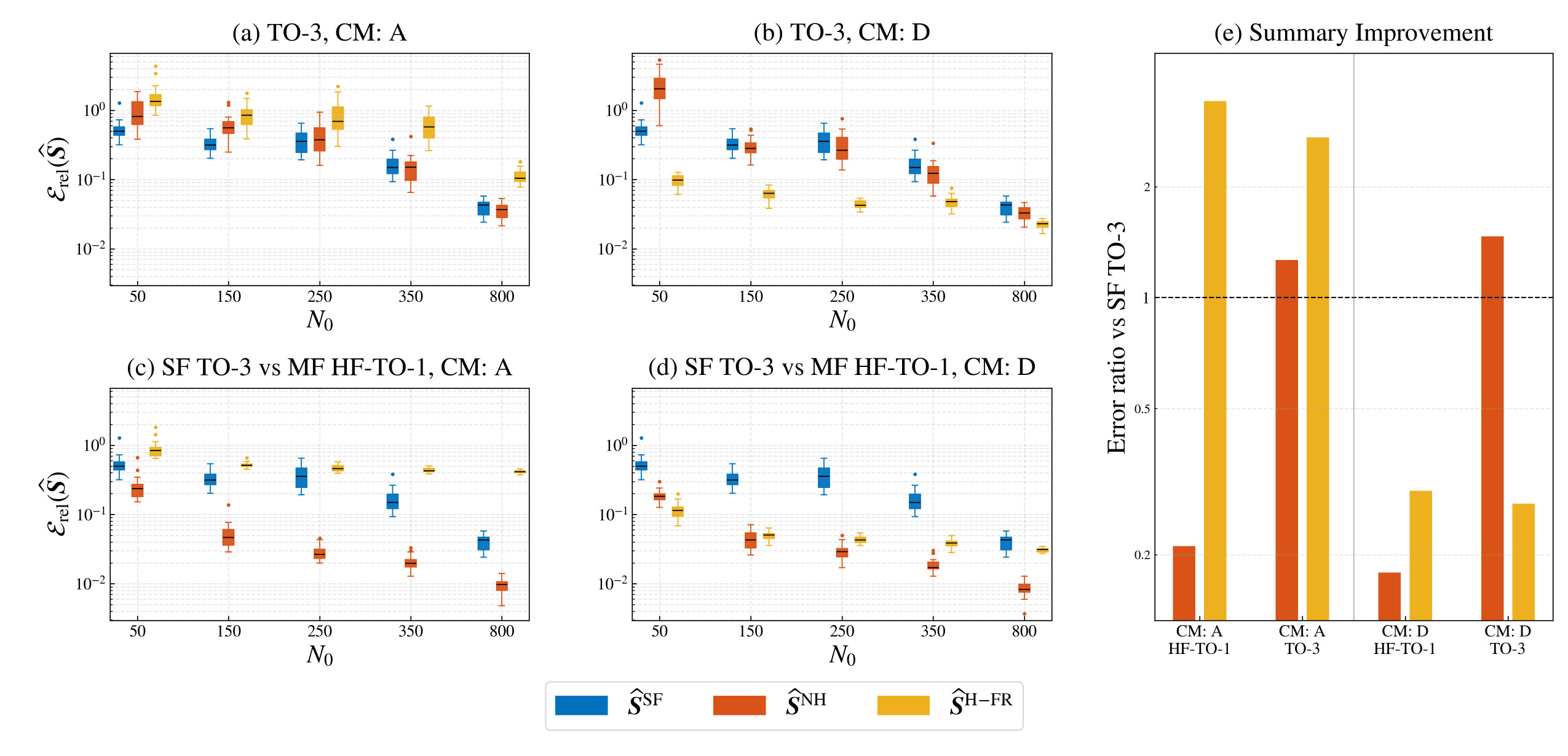}
    \caption{Quantitative comparisons of transport maps in Section~\ref{sec:nealfunnel_example}. Panels (a,b) compare TO-3 maps under cost models A and D. Panels (c,d) compare the single-fidelity TO-3 baseline with multifidelity maps using HF-TO-1, again under cost models A and D. Panel (e) summarizes median error ratios relative to the single-fidelity TO-4 baseline, averaged over all $N_0$; values below one indicate improvement over the baseline.}
    \label{fig:neal_funnel_overall_results}
\end{figure}

Figure~\ref{fig:neal_funnel_overall_results} mirrors the comparison in Figure~\ref{fig:quartic_banana_overall_results}, but in a higher-dimensional setting where the true KR map is non-polynomial. Panels (a,b) use a common TO-3 parameterization. Unlike the TO-4 comparison in Section~\ref{sec:banana_example}, this parameterization is not an exactly specified map class; rather, it was chosen as the minimum TO class that empirically enabled an accurate approximation of the true high-fidelity KR map in a single-fidelity setting with plentiful training data. 

The results under this generic parameterization are less uniformly favorable to multifidelity transport than in the previous example. The hierarchical method is strongly hierarchy-dependent: under cost model A, where the hierarchy places the misaligned source adjacent to $\pi_0$, it performs worse than both $\SSF$ and $\SNH$, while under cost model D, where the aligned source is adjacent to $\pi_0$ and more low-fidelity information is available overall, it can substantially reduce error at small and intermediate values of $N_0$. This contrast reinforces the hierarchy dependence observed in Section~\ref{sec:banana_example}, but the effect is more pronounced here. In the higher-dimensional, non-polynomial funnel, an unfavorable lower-fidelity cumulative map can leave a final residual that is more difficult to learn than the original high-fidelity map.

The comparison also reinforces the previous discussion concerning the distinction between \emph{explicit} and \emph{implicit} use of lower-fidelity information. When the hierarchy is favorable, H-FR explicitly applies a lower-fidelity transport map before training the final high-fidelity stage, so the final map only needs to learn a residual correction. This can make the hierarchical method robust to a generic final-stage parameterization. In contrast, $\SNH$ uses low-fidelity information \emph{implicitly} through monotone corrections to its parameterization. Under a generic TO-3 parameterization, the high-fidelity shift remains sufficiently expressive that the fitted map is not strongly forced to rely on the available low-fidelity structure that offers a potential reduction in sample complexity. 

Panels (c,d) of Figure~\ref{fig:neal_funnel_overall_results} test the high-fidelity complexity-reduction mechanism more directly. The single-fidelity baseline retains a TO-3 parameterization, while the multifidelity methods use HF-TO-1. Thus, as in Section~\ref{sec:banana_example}, the part of the map learned directly from scarce high-fidelity samples is restricted, while the low-fidelity maps retain the higher-order structure. Under this comparison, $\SNH$ gives the clearest and most stable improvement over the single-fidelity baseline. This indicates that the corrected low-fidelity components and trainable scales supply nonlinear structure that would otherwise need to be learned from high-fidelity samples alone.

The hierarchical method remains competitive only under a favorable ordering. Under cost model D, where $\pi_2$ is adjacent to $\pi_0$, $\SFR$ improves over the single-fidelity baseline in several regimes. Under cost model A, the same method degrades because the final high-fidelity stage must correct the effect of the misaligned source. Panel (e) summarizes these trends: the best aggregate performance is obtained by $\SNH$ with HF-TO-1, while $\SFR$ is effective only when the hierarchy is well aligned with the target geometry.
\begin{figure}[!htb]
\centering
\includegraphics[clip, trim={0 0.1cm 0 0}, width=0.95\linewidth]{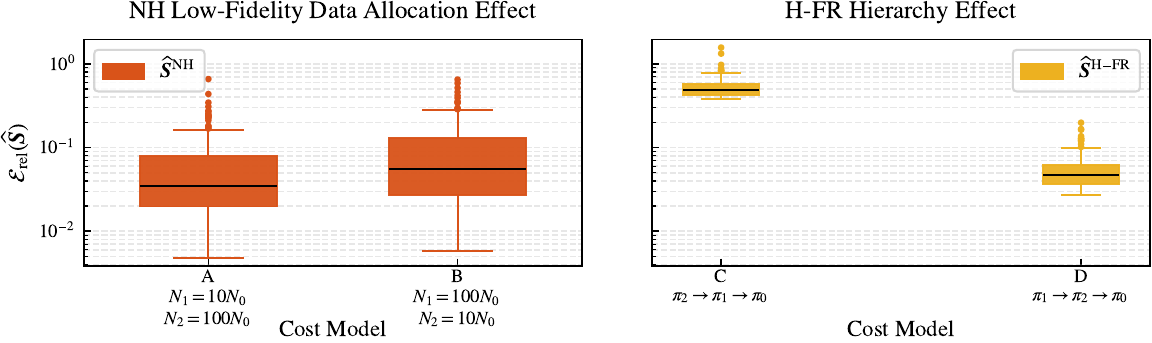}
\caption{Sensitivity to low-fidelity sample allocation and hierarchical ordering for HF-TO-1 multifidelity maps in Section~\ref{sec:nealfunnel_example} aggregated over $N_0\in\{50,100,150,250, 350, 800\}$. Left: $\SNH$ under cost models A and B, comparing under different budgets for $\pi_1$ and $\pi_2$. Right: $\SFR$ under cost models C and D, which use equal low-fidelity sample sizes but reverse the hierarchy order.}
\label{fig:nealfunnel_hierarchy_and_allocation_sensitivity}
\end{figure}

Figure~\ref{fig:nealfunnel_hierarchy_and_allocation_sensitivity} isolates the allocation and hierarchy effects using the cost model interpretation from Table~\ref{tab:numerical_cost_models}. The left panel compares $\SNH$ under cost models A and B, thereby isolating the effect of exchanging the low-fidelity sample allocation between the misaligned source $\pi_1$ and the aligned source $\pi_2$. The aggregate differences are modest relative to the variation across training replicates, indicating that the non-hierarchical construction is not highly sensitive to this allocation change.

The right panel shows a much stronger hierarchy effect for $\SFR$. Cost models C and D use the same low-fidelity sample sizes but reverse the hierarchy. The hierarchy in cost model D performs substantially better than the reversed hierarchy in cost model C, again showing that the fixed-reference hierarchy is effective only when the lower-fidelity cumulative map leaves a simple high-fidelity residual.

Overall, the softened Neal-funnel example reinforces the main conclusion of Section~\ref{sec:banana_example} in a more demanding setting. Multifidelity information is most useful when it reduces the complexity of the map component that must be learned from scarce high-fidelity samples. The non-hierarchical method benefits most from this mechanism when the high-fidelity shift is restricted, giving the most stable improvement in the HF-TO-1 comparison. The fixed-reference hierarchical method can also be effective, and can be more robust to a generic final-stage parameterization under a favorable ordering, but the higher dimension and non-polynomial funnel geometry make poor hierarchy choices substantially more detrimental to the overall performance.

\subsection{Multifidelity simulation-based inference}
\label{sec:LV}
For our final example, we consider a simulation-based inference (SBI) problem involving a stochastic Lotka--Volterra (LV) predator-prey model, which is often used as a benchmark in SBI and other Bayesian inference tasks~\cite{martinBayesianModelingComputation2021,wang2025efficient,wilkinson2018stochastic,lueckmanBenchmarkingSBI2021}. This example differs from the preceding tests in two respects. First, the high-fidelity target density is not available in closed form, and the corresponding high-fidelity KR map is unknown. Second, the object of interest is a conditional distribution over model parameters given summary statistics of an observed trajectory, rather than an unconditional target distribution. The experiment therefore tests the impact that multifidelity triangular transport can make in improving conditional inference, when exact stochastic simulations are scarce. %

We introduce here a multifidelity extension of the LV conditional inference problem from \cite{papamakarios2016fast}. The model has four positive reaction rates, denoted by $\bm{x}\in\mathbb{R}_{+}^{4}$, and the simulator output is reduced to a summary statistic $\bm{y}\in\mathbb{R}^9$. We write $u_i=\log x_i$ and use the physical prior $\bm u\sim p=\bigotimes_{i=1}^{4}\mathcal{U}(-5,2)$. For transport training and conditional sampling, we apply the fixed componentwise CDF transform from the standard normal CDF.
Each sample is represented by the joint input-output variable
$\bm{w}=(\bm{y},\bm{\xi})\in\mathbb{R}^{13}$, and we define three fidelity-indexed joint distributions $\pi_\ell(\bm w)$, for $\ell=0,1,2$.

The high-fidelity distribution $\pi_0$ is generated using the Gillespie stochastic simulation algorithm (SSA)~\cite{gillespie1977exact}. The first lower-fidelity distribution $\pi_1$ uses a fixed-step tau-leaping approximation~\cite{gillespie2001approximate}, and the second lower-fidelity distribution $\pi_2$ uses a deterministic mean-field ODE approximation~\cite{kurtz1970solutions}. We note that, unlike the often-used ODE example, the SSA solver produces a stochastic output for a fixed input; as such, this is truly a likelihood-free inference problem and thus we cannot create reference samples without resorting to, e.g., approximate Bayesian computation (ABC)~\cite{martinBayesianModelingComputation2021}. Figure~\ref{fig:LV_fidelities} illustrates trajectories from the three simulators for two different reaction-rate vectors.
\begin{figure}[!htb]
    \centering
    \includegraphics[clip, trim={0 0.45cm 0 0}, width=0.9\linewidth]{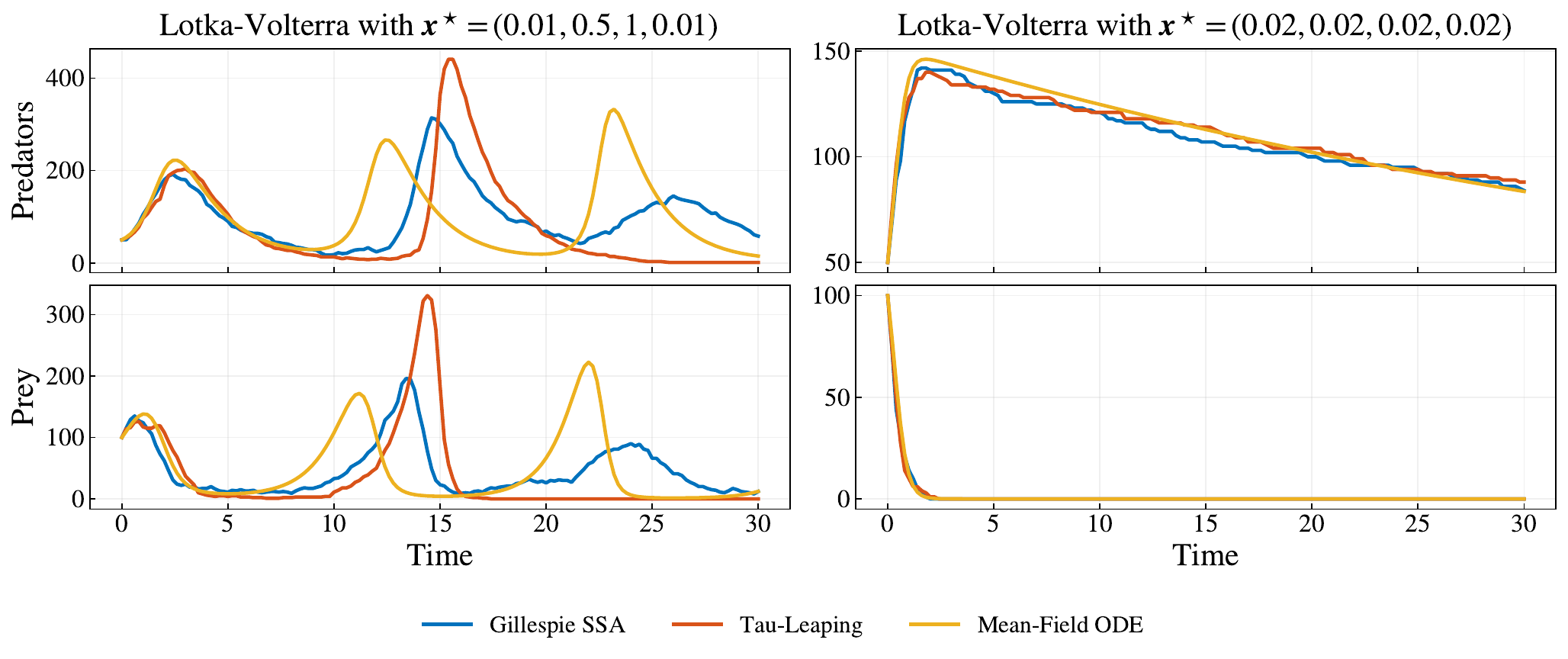}
    \caption{Example predator-prey trajectories in Section~\ref{sec:LV}, generated by Gillespie SSA, tau-leaping, and a deterministic mean-field ODE approximation for two different reaction-rate vectors.}
    \label{fig:LV_fidelities}
\end{figure}
The ODE model is the cheapest source and removes intrinsic stochasticity, while tau-leaping retains stochastic reaction counts at reduced computational cost relative to SSA. Thus, unlike the analytical examples in Sections~\ref{sec:banana_example} and~\ref{sec:nealfunnel_example}, the lower-fidelity sources are not prescribed geometric perturbations of the high-fidelity density. They arise from different approximations of the stochastic reaction process, and their usefulness for high-fidelity conditional inference must be learned from samples. To keep each joint distribution nondegenerate in the summary variables, we use a Gaussian noise model for the summaries at every fidelity. Details on the reaction network, transport-coordinate transformation, simulator implementations, summary statistics, noise model, SMC--ABC reference posterior, and metric estimators are collected in Appendix~\ref{app:lv-details}.

We aim to approximate the high-fidelity conditional distribution $\pi_0(\bm u\mid \bm y^\star)$, where $\bm y^\star$ is an observed summary generated from an SSA trajectory. Triangular transport provides an amortized approach to this task using samples from the transformed joint distribution $\pi_0(\bm w)$. Specifically, we approximate the high-fidelity KR map $\widehat{\bm S}\approx \SZero:\mathbb{R}^{13}\rightarrow\mathbb{R}^{13}$ satisfying $(\SZero)_\sharp\pi_0=\eta$, where $\eta=\mathcal{N}(0,I_{13})$. Under the ordering $\bm w=(\bm y,\bm \xi)$, the learned map decomposes as
\begin{equation}
\label{eqn:LV_conditional_decomposition}
    \widehat{\bm{S}}(\bm{y},\bm{\xi})=
    \begin{bmatrix}
        \widehat{\bm{S}}_{\bm{y}}(\bm{y})\\
        \widehat{\bm{S}}_{\bm{\xi}}(\bm{y},\bm{\xi})
    \end{bmatrix},\quad \widehat{\bm{S}}_{\bm{u}}(\bm{y},\bm{u}) \coloneq \widehat{\bm{S}}_{\bm{\xi}}(\bm{y},\Psi(\bm{u})).
\end{equation}
For a fixed observed $\bm y^\star$, conditional inference only requires the second block $\widehat{\bm S}_{\bm\xi}$. We draw $\bm z \sim \eta_{\bm\xi}= \mathcal{N}(0,I_4)$, solve the system $\widehat{\bm S}_{\bm \xi}(\bm y^\star,\bm \xi)=\bm z$ for $\bm \xi$, then transform back to log-rate coordinates $\bm u$. %
The resulting pullback on $\bm u$ is
\begin{equation}
\label{eqn:LV_conditional_density}
    \widehat{\pi}_{\bullet}(\bm u\mid \bm y^\star)
    =
        \left(\widehat{\bm S}^{\bullet}_{\bm u}(\bm y^\star,\cdot)\right)^{\sharp}
        \eta_{\bm \xi}
\end{equation}
with $\widehat{\bm S}^{\bullet}_{\bm{u}}$ as the bottom block of $\SNH$, $\SFR$, or $\SSF$ composed with $\Psi$ as in~\eqref{eqn:LV_conditional_decomposition}.

We vary the high-fidelity sample size over $N_0\in\{100,500,1000,2000\}$ and use the two fixed low-fidelity cost models in Table~\ref{tab:LV_cost_models}. In contrast to the previous examples, the low-fidelity sample sizes are not multiples of $N_0$; they are kept fixed as $N_0$ increases. This reflects an SBI setting in which a reusable pool of lower-fidelity simulations is available and the main bottleneck is the number of true SSA simulations.
\begin{table}[!htb]
\centering
\caption{Low-fidelity sample sizes and hierarchical orderings used in Section~\ref{sec:LV}.}
\label{tab:LV_cost_models}

\begin{tabular}{@{}c c c c@{}}
\toprule
Cost model & H-MF hierarchy & $N_1$ & $N_2$ \\
\midrule
$\mathrm{C}_{\mathrm{flf}}$ & $\pi_2 \to \pi_1 \to \pi_0$ & $10000$ & $10000$ \\
$\mathrm{D}_{\mathrm{flf}}$ & $\pi_1 \to \pi_2 \to \pi_0$ & $10000$ & $10000$ \\
\bottomrule
\end{tabular}
\end{table}
Cost model $\mathrm{C}_{\mathrm{flf}}$ thus uses ODE $\to$ tau leaping $\to$ SSA, while $\mathrm{D}_{\mathrm{flf}}$ reverses the two lower-fidelities. %
The non-hierarchical method is unaffected due to identical sample sizes $N_1 = N_2$.%

As a baseline, we choose total order 3 parameterizations, which we found produces satisfactory results given plentiful high-fidelity data. Thus, we compare a single-fidelity TO-3 map and with multifidelity maps using both TO-3 and HF-TO-1 parameterizations. Consistent with previous examples, we select the regularization parameter by the cross-validation procedure described in Appendix~\ref{app:regularization_cv}.
\begin{figure}[!htb]
    \centering
    \includegraphics[width=\linewidth]{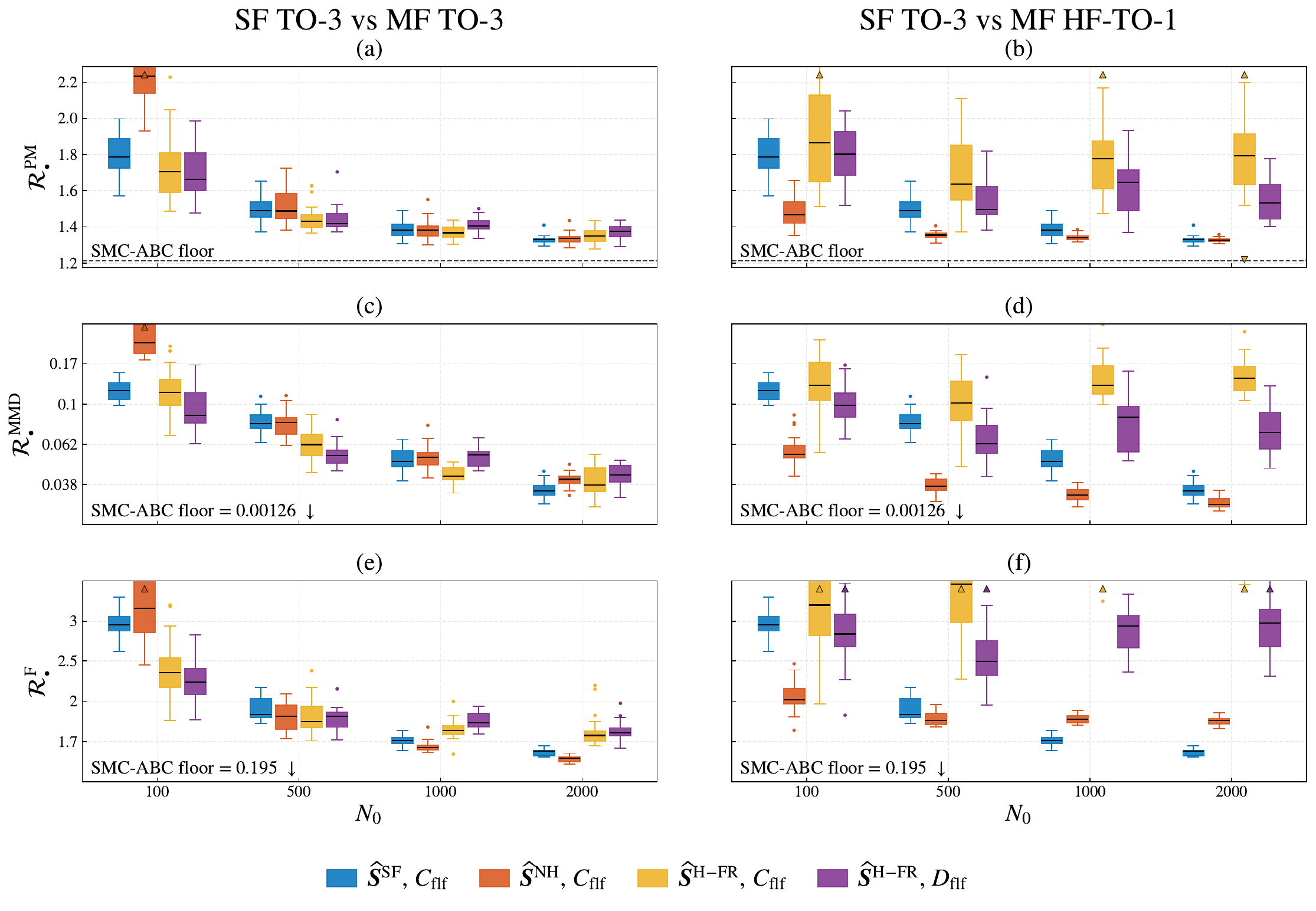}
    \caption{Amortized conditional risk metrics \eqref{eqn:lv_prior_predictive_risk} for Section~\ref{sec:LV}. From top to bottom the rows show posterior mean, MMD, and F\"orstner risks. The left column compares the maps under a common TO-3 parameterization. The right column compares the single-fidelity TO-3 baseline to the multifidelity HF-TO-1 parameterizations. Each boxplot summarizes 20 replications on 20 realizations of training data. The floor in each panel is the error for each risk between two independent SMC-ABC datasets of the same size used to estimate the respective risk of each method.}
    \label{fig:lv_quantitative}
\end{figure}

Since the high-fidelity KR map is unknown in this example, we evaluate amortized
conditional inference accuracy using prior-predictive risk metrics:

\begin{align}
\hspace{7em}\mathcal{R}_{\bullet}^{\mathfrak{m}}
    \coloneq
    \mathbb{E}&_{\bm u\sim p,\; \bm y\sim \pi_0(\cdot\mid \bm u)}
    \big[
        \Delta_{\bullet}^{\mathfrak{m}}(\bm u,\bm y)
    \big],
    \qquad
    \mathfrak{m}\in \{\mathrm{PM},\mathrm{MMD},\mathrm{F}\},
    \label{eqn:lv_prior_predictive_risk}\\[0.3em]
\Delta_{\bullet}^{\mathrm{PM}}(\bm u,\bm y)
&:=
\left\|
    \bm u
    -
    \mathbb{E}_{\bm z\sim \eta_{\bm \xi}}
    \left[
    \left(\widehat{\bm{S}}^\bullet_{\bm{u}}(\bm{y},\cdot)\right)^{-1}(\bm{z})
    \right]
\right\|_2^2,
\label{eqn:lv_pm_loss}
\\[0.5em]
\Delta_{\bullet}^{\mathrm{MMD}}(\bm u,\bm y)
&:=
\mathrm{MMD}_{\mathrm{ms}}^2
\left(
    \widehat{\pi}_{\bullet}(\cdot\mid \bm y),
    \widehat{\pi}_{\mathrm{ABC}}(\cdot\mid \bm y)
\right),
\label{eqn:lv_mmd_loss}
\\[0.5em]
\Delta_{\bullet}^{\mathrm{F}}(\bm u,\bm y)
&:=
d_{\mathrm{F}}
\left(
    \widehat{\pi}_{\bullet}(\cdot\mid \bm y),
    \widehat{\pi}_{\mathrm{ABC}}(\cdot\mid \bm y)
\right),
\label{eqn:lv_forstner_loss}
\end{align}
where only the posterior mean integrand $\Delta^{\mathrm{PM}}_\bullet$ depends on the drawn value of log-rates $\bm{u}$. %
The metric $\mathcal{R}_{\bullet}^{\mathrm{PM}}$ measures squared error between the transport-induced posterior mean and the true generating parameter. The metric $\mathcal{R}_{\bullet}^{\mathrm{MMD}}$ measures distributional discrepancy relative to the reference posterior $\widehat{\pi}_{\mathrm{ABC}}(\cdot\mid \bm y)$ using a multiscale squared maximum mean discrepancy $\mathrm{MMD}_{\mathrm{ms}}$. In contrast,
$\mathcal{R}_{\bullet}^{\mathrm{F}}$ compares the posterior covariance structure of $\widehat{\pi}_{\bullet}(\cdot\mid \bm y)$ and $\widehat{\pi}_{\mathrm{ABC}}(\cdot\mid \bm y)$ using the F\"orstner distance \cite{forstner2003metric}. Both metrics are estimated from samples generated by the transport map and reference samples generated by the PyMC SMC-ABC sampler~\cite{pymc2023}, which implements the methodology of~\cite{toni2008approximate}. Additional details on the numerical approximation of these metrics are given in Appendix~\ref{app:lv-details}.

Figure~\ref{fig:lv_quantitative} reports the three amortized conditional risk metrics defined in \eqref{eqn:lv_prior_predictive_risk}. The main trend is consistent with the preceding examples. Under a generic TO-3 parameterization, the non-hierarchical method does not reliably improve over the single-fidelity TO-3 baseline. In contrast, the NH HF-TO-1 parameterization gives the clearest and most consistent improvement when high-fidelity training data are scarce. This contrast suggests that low-fidelity maps contain nonlinear structure useful for approximating the high-fidelity map, where this structure is most readily exploited when the expressivity of the high-fidelity shift function is limited.

The improvement over the single-fidelity baseline is strongest for the MMD, which captures features of the entire distribution, and posterior mean risk metrics. The F\"orstner risk, which only measures the error in conditional covariance, shows a more nuanced behavior. The non-hierarchical approach of NH HF-TO-1 improves over the single-fidelity baseline at the smallest value of $N_0$, but then stagnates above the
single-fidelity error as $N_0$ increases. This suggests that the low-order high-fidelity shift restriction, while effective in allocating posterior mass to the correct regions and reducing MMD risk, can limit the covariance structure recovered by the transport-induced posterior. This is consistent with the MMD results: the MMD and F\"orstner metrics emphasize
different aspects of the posterior approximation. %
This expressivity-based diagnosis is further supported by the NH TO-3 results. Under this parameterization which uses a total order three high-fidelity shift function, we do not observe a stagnating F\"orstner risk relative to the single-fidelity baseline for large $N_0$.
\begin{figure}[!htb]
    \centering
    \includegraphics[width=\linewidth]{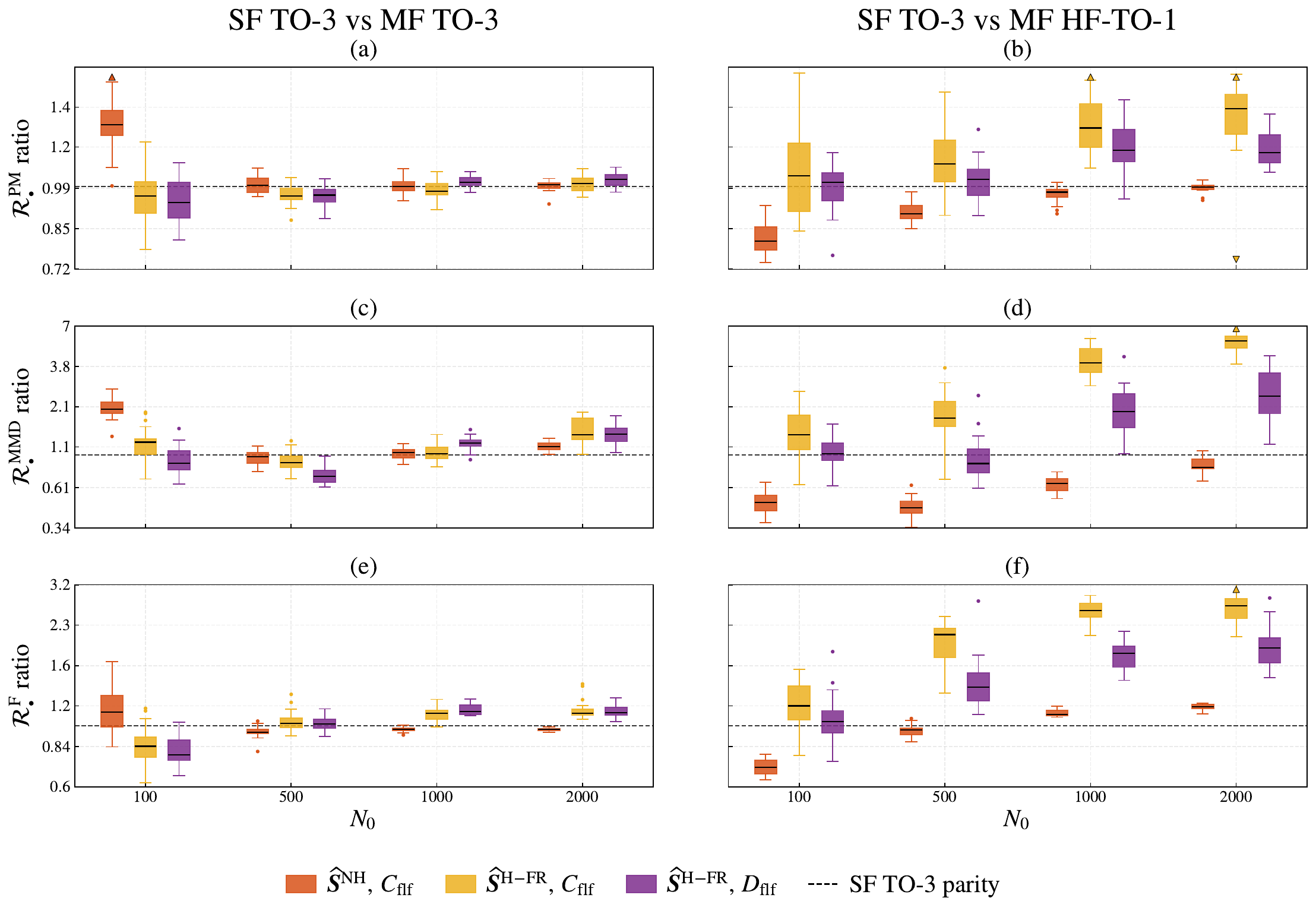}
    \caption{Amortized risk ratios for multifidelity methods relative to the single-fidelity TO-3 baseline over the 20 training replications. From top to bottom the rows present results for posterior mean, MMD, and F\"orstner risks. The left column compares single-fidelity, non-hierarchical, and hierarchical fixed reference maps under a common TO-3 parameterization. The right column compares the single-fidelity TO-3 baseline to NH and H-FR maps using a HF-TO-1 parameterization.}
    \label{fig:lv_quantitative_ratios}
\end{figure}

The H-FR maps do not exhibit the same robust improvement over single-fidelity transport. Unlike the favorable hierarchy regimes in Sections~\ref{sec:banana_example} and~\ref{sec:nealfunnel_example}, neither LV hierarchy gives a convincing advantage across the
three risk metrics. Moreover, H-FR generally performs better with a TO-3 final stage than under the HF-TO-1 parameterization, which uses a TO-1 final stage. This indicates that the residual correction from the lower-fidelity composition to the SSA target retains non-negligible higher-order structure. There is still a hierarchy effect: $D_{\mathrm{flf}}$ is generally better than
$C_{\mathrm{flf}}$, especially for HF-TO-1. Cost model $C_{\mathrm{flf}}$, however, corresponds to the natural simulator ordering ODE $\to$ tau-leaping $\to$ SSA, suggesting that simulator fidelity may not align with the hierarchy that produces the simplest transport problem.
\begin{figure}[!htb]
    \centering
    \includegraphics[clip, trim={0 0.3cm 0 0}, width=0.95\linewidth]{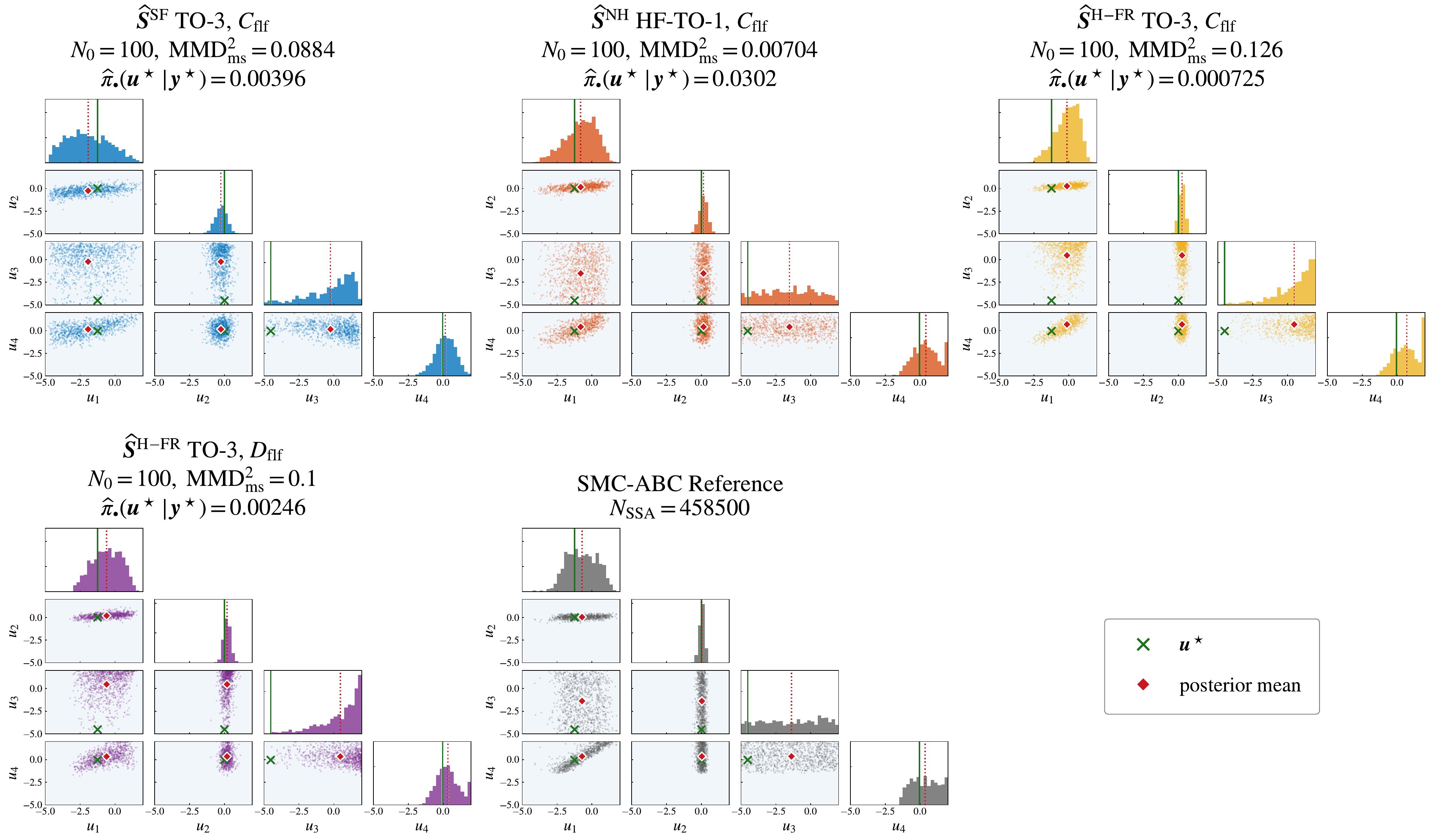}
    \caption{Corner plot showing samples from $\widehat{\pi}_{\bullet}(\bm u \mid \bm y^\star)$. Single-fidelity samples are generated by the median performing map $\SSF$ with respect to $\mathcal{R}^{\mathrm{MMD}}_{\bullet}$ at $N_0=100$. Multifidelity comparison maps use identical high-fidelity training samples to the single-fidelity baseline. The observation $\bm y^\star$ is a vector of summary statistics from a single SSA trajectory for parameter $\bm u^\star = \log \bm x^\star$ with $\bm x^\star=(0.29, 1.0, 0.01, 0.95)$. Titles indicate the map type, parameterization, and cost model. We also show reference SMC-ABC samples using $N_{\mathrm{SSA}}=458500$ high-fidelity SSA simulations.}
    \label{fig:balanced_corners_N0_100}
\end{figure}

Figure~\ref{fig:lv_quantitative_ratios} provides a complementary view on the same results. The boxplots summarize error ratios between the multifidelity methods and the single-fidelity TO-3 baseline paired by training replicate. Importantly, for a fixed training trial, the single-fidelity and multifidelity maps use identical high-fidelity training samples. These results therefore address the question: for a single set of high- and low-fidelity training data, how often and by how much do we expect the multifidelity methods to improve over the single-fidelity baseline? This paired view is even more favorable to NH HF-TO-1: for posterior mean and MMD risk, it consistently improves on the single-fidelity baseline in the sparse-data trials, with the advantage shrinking as $N_0$ increases. 

Finally, Figure~\ref{fig:balanced_corners_N0_100} gives a qualitative view of one conditional posterior at $N_0=100$. It shows samples generated by the median-performing single-fidelity map with respect to MMD risk, and it compares these with samples generated by multifidelity maps trained on the identical high-fidelity training set. This single target should not be overinterpreted in an amortized experiment, but it is consistent with the aggregate metrics: the single-fidelity posterior is comparatively diffuse, NH HF-TO-1 gives the most credible agreement with the SMC-ABC reference among the displayed maps, and the H-FR TO-3 maps show heterogeneous behavior across marginals.

\section{Conclusion}\label{sec:conclusion}

While triangular transport maps are broadly useful in Bayesian computation and uncertainty quantification, accurate map construction can still require more high-fidelity samples than is practical. We address this limitation by proposing and exploring two multifidelity triangular transport constructions that combine scarce high-fidelity data with more abundant low-fidelity data. Our \textit{hierarchical} approach composes maps along a user-specified fidelity ordering, with changing-reference and fixed-reference variants that differ in composition order and optimization properties. Our \textit{non-hierarchical} approach avoids a prescribed ordering by representing the high-fidelity transport through low-complexity scalings and shifts of pretrained lower-fidelity maps.

Several numerical experiments in Section~\ref{sec:numerics} show that the value of multifidelity information is not automatic; it depends on whether lower-fidelity structure reduces the complexity of what must be learned from scarce high-fidelity samples. This principle appears consistently across the polynomial verification example in Section~\ref{sec:banana_example}, the higher-dimensional Neal's funnel example in Section~\ref{sec:nealfunnel_example}, and the Lotka--Volterra simulation-based inference example in Section~\ref{sec:LV}.

The non-hierarchical method works best when paired with a parameterization that encourages reliance on the pretrained low-fidelity maps. With a low-order high-fidelity shift, for instance, the low-fidelity data provide the nonlinear transport structure, and the high-fidelity samples are used mainly to learn low-complexity corrections and scale coefficients. The scale coefficients also provide a mechanism for suppressing less useful information sources, making the method insensitive to the chosen fidelity hierarchy. Hierarchical methods, by contrast, depend directly on the choice of hierarchy. When the ordering is favorable, the staged transformations can reduce the complexity of the final stage map and make the method robust to a na\"ive parameterization. When the ordering is poor, however, errors propagate through the composition and make later stages harder to learn. In our experiments, the fixed-reference variant is generally faster and more stable to train than the changing-reference variant, which casts aside the standard Gaussian reference after the first stage.

There are several compelling directions for future work. Cost-aware adaptive methods could select both sample allocations and map expressivity across fidelities in an automated manner. Such adaptive parameterizations may also improve scalability by exploiting sparse conditional dependence, as has been demonstrated for triangular maps in the single-fidelity setting~\cite{lunde2026adaptive}; however, these ideas have yet to be developed in the context of multifidelity data. 
Alternative parameterizations, such as neural networks, radial basis functions, splines, and other adaptive bases may further improve scalability or localization~\cite{bryutkin2025neural,lunde2026adaptive}. We also wish to extend the framework to settings where multifidelity information is available through unnormalized densities rather than samples alone~\cite{sharp2026sampling,muchandimath2025accelerating}.
Another direction of interest is the extension of multifidelity transport to federated settings, where different information sources remain on separate clients and only gradient or parameter updates can be communicated. More broadly, the hierarchical approaches developed in this work are readily transferrable to other transport representations, such as normalizing flows and flow matching. By contrast, the non-hierarchical approach is tailored to the present triangular transport setting, and the development of non-hierarchical approaches for other transport methods remains a compelling open research direction.

\section*{Acknowledgments}
This material is based upon work supported by the U.S. Department of Energy, Office of Science, Office of Advanced Scientific Computing Research. YM also acknowledges support from the DOE Office of Science, Office of Advanced Scientific Computing Research (ASCR), via the FASTMath6 SciDAC Institute (contract number DE-AC52-07NA27344, under a subcontract from Lawrence Livermore National Laboratory), via the M2dt MMICC center (award number DE-SC0023187), and via the ROADMAP project (award number DE-SC0026245).

Sandia National Laboratories is a multi-mission laboratory managed and operated by National Technology \& Engineering Solutions of Sandia, LLC (NTESS), a wholly owned subsidiary of Honeywell International Inc., for the US Department of Energy’s National Nuclear Security Administration (DOE/NNSA) under contract DE-NA0003525. This written work is authored by an employee of NTESS. The employee, not NTESS, owns the right, title and interest in and to the written work and is responsible for its contents. Any subjective views or opinions that might be expressed in the written work do not necessarily represent the views of the U.S. Government. The publisher acknowledges that the U.S. Government retains a non-exclusive, paid-up, irrevocable, world-wide license to publish or reproduce the published form of this written work or allow others to do so, for U.S. Government purposes. The DOE will provide public access to results of federally sponsored research in accordance with the DOE Public Access Plan.

\bibliographystyle{plain}
\bibliography{references}

\clearpage
\appendix
\setcounter{section}{0}
\setcounter{equation}{0}
\setcounter{figure}{0}
\setcounter{table}{0}
\renewcommand{\thesection}{\Alph{section}}
\renewcommand{\thesubsection}{\thesection.\arabic{subsection}}
\renewcommand{\theequation}{\thesection.\arabic{equation}}
\renewcommand{\thefigure}{\thesection.\arabic{figure}}
\renewcommand{\thetable}{\thesection.\arabic{table}}
\renewcommand{\theHsection}{appendix.section.\Alph{section}}
\renewcommand{\theHequation}{appendix.equation.\Alph{section}.\arabic{equation}}
\renewcommand{\theHfigure}{appendix.figure.\Alph{section}.\arabic{figure}}
\renewcommand{\theHtable}{appendix.table.\Alph{section}.\arabic{table}}

\section*{Appendices}
\addcontentsline{toc}{section}{Appendices}

\section{Additional Details on Triangular Transport}\label{app:triangular_transport_details}

This appendix collects triangular transport details that are useful for implementation but not essential to the main development in Section~\ref{sec:triangular_transport}.

\subsection{Map inversion}\label{app:map_inversion}

As noted in Section~\ref{sec:triangular_transport}, triangularity makes inversion recursive. Let $\bm{S}$ be a monotone triangular map of the form \eqref{eqn:tri_transport_structure}. To evaluate $\bm{x}^*=\bm{S}^{-1}(\bm{z}^*)$ for a given $\bm{z}^*\in\mathbb{R}^d$, one solves
\begin{equation}\label{eqn:map_inversion_app}
    \bm{S}(\bm{x}^*)
    =
    \begin{bmatrix}
        S_1(x_1^*)\\
        S_2(x_1^*,x_2^*)\\
        \vdots\\
        S_d(x_1^*,\ldots,x_d^*)
    \end{bmatrix}
    =
    \begin{bmatrix}
        z_1^*\\
        z_2^*\\
        \vdots\\
        z_d^*
    \end{bmatrix}.
\end{equation}
We find $x_k^*$ recursively. Suppose $x_1^*,\ldots,x_{k-1}^*$ have been computed; then, define the one-dimensional function
\begin{equation}\label{eqn:scalar_inversion_function_app}
    \varphi_k(t;x_1^*,\ldots,x_{k-1}^*)
    :=
    S_k(x_1^*,\ldots,x_{k-1}^*,t).
\end{equation}
The $k$-th coordinate can thus be obtained from
\begin{equation}\label{eqn:inversion_recursion_app}
    x_k^*
    =
    \varphi_k(\,\cdot\,;x_1^*,\ldots,x_{k-1}^*)^{-1}(z_k^*),
    \qquad
    k=1,\ldots,d .
\end{equation}
Evaluating $\bm{S}^{-1}$ subsequently requires solving $d$ scalar root-finding problems. If $\partial_k S_k>0$, any solution to this scalar problem is unique~\cite{Marzouk2016}.

\subsection{Additional information about map component parameterization}\label{app:map_param}

In the numerical experiments, each function $f_k$ appearing in \eqref{eqn:integrated_rectifier} is represented by a finite expansion

\begin{equation}\label{eqn:polynomial_parameterization}
    f_k(\bm{x}_{\leq k};\bm{\theta}_k)
    =
    \sum_{\bm{\alpha}\in\Lambda_k}
    c_{\bm{\alpha}}\Phi_{\bm{\alpha}}(\bm{x}_{\leq k}),
    \qquad
    \bm{\theta}_k
    =
    \{c_{\bm{\alpha}}:\bm{\alpha}\in\Lambda_k\},
\end{equation}
where $\Lambda_k$ is a finite set of $k$-dimensional multi-indices and $\Phi_{\bm{\alpha}}$ is a tensor product of one-dimensional basis functions \cite{le2010spectral,ernst2012convergence}. A \emph{total-order} index set of order $p$ is

\begin{equation}\label{eqn:total_order_index_set}
    \Lambda_k(p)
    :=
    \left\{
        \bm{\alpha}\in\mathbb{N}_0^k:
        |\bm{\alpha}|_1
        =
        \sum_{r=1}^k \alpha_r
        \leq p
    \right\}.
\end{equation}
This is the convention used later when referring to total-order map parameterizations.

Following \cite{ramgraber2025friendly}, we use a combination of probabilists' Hermite polynomials $\{\mathrm{He}_n\}_{n=0}^{\infty}$ and Hermite functions
\begin{equation}\label{eqn:Hermite_functions}
    \mathcal{H}_n(x)
    =
    \mathrm{He}_n(x)\exp\!\left(-\frac{x^2}{\sigma}\right),
\end{equation}
where $n$ denotes polynomial degree. Hermite functions multiply the classical Hermite polynomial by an exponentially decaying factor, which damps the basis far from the origin. The damping is controlled by the parameter $\sigma$, which is fixed at $30$ in this work. This edge control improves stability when evaluating or inverting maps outside the training region, although excessive damping can reduce extrapolation accuracy. In the experiments, we use Hermite functions for terms of order two and higher, while retaining standard Hermite polynomials for constant and linear terms.

The positive rectifier $g$ in \eqref{eqn:integrated_rectifier} enforces monotonicity through \eqref{eqn:integrated_rectifier_derivative}. Different rectifiers can affect parameter identifiability and optimization behavior \cite{baptista2024representation}. We use the softplus rectifier

\begin{equation}\label{eqn:rectifier}
    g(x)
    =
    \mathrm{SoftPlus}(x)
    =
    \log(1+\exp(x)),
\end{equation}
which is smooth and strictly positive, and has been used effectively in prior triangular transport constructions \cite{baptista2024representation,ramgraber2025friendly}.

\subsection{Optimization details}
\label{app:optimization_details}

The finite-dimensional optimization problems arising in map training are solved using the BFGS and limited-memory BFGS (L-BFGS) routines implemented in \texttt{SciPy} \cite{virtanen2020scipy}. These routines are used throughout the numerical experiments for single-fidelity, hierarchical multifidelity, and non-hierarchical multifidelity map training.

\section{Hierarchical Transport and Deep Maps}
\label{app:hierarchical_transport_details}

This appendix expands Section~\ref{sec:hierarchical_multifidelity} by relating the hierarchical multifidelity constructions to single-fidelity deep transport maps. Compositions of triangular maps have been used to improve transport approximations without substantially enriching each individual map parameterization \cite{el2012bayesian,Marzouk2016,parno2018transport,peherstorfer2019transport}. In particular, the authors in \cite{peherstorfer2019transport} describe a deep-map construction based on composing maps in coarse approximation spaces and note an alternative construction that reverses the composition direction by keeping the reference fixed and pulling back the target through the current map approximation \cite{peherstorfer2019transport}; see also the data-driven density-estimation perspective of \cite{tabak2013family}. Related deep compositions of transport maps have also been developed using tensor-train representations \cite{cui2022deep}.

\subsection{Single-fidelity changing-reference composition}
\label{app:single_fidelity_changing_reference}

Consider a single-fidelity problem of measure transport with target density $\pi$, reference density $\eta$, and samples $\{\bm{x}^{(j)}\}_{j=1}^{N}\sim\pi$. The goal is to approximate the KR map $\bm{S}$ satisfying $\bm{S}_{\sharp}\pi=\eta$.

In a changing-reference composition, one constructs a sequence of triangular maps
\[
    \widehat{\bm{A}}_L,\widehat{\bm{A}}_{L-1},\ldots,\widehat{\bm{A}}_1 ,
\]
where each stage uses the same target samples but a reference density induced by the previous stage. This indexing is chosen to parallel the changing-reference multifidelity construction in Section~\ref{sec:hcr_formulation}: the composition is trained from $m=L$ down to $m=1$, with the prescribed reference appearing at the first stage.

Set
\begin{equation}\label{eqn:sf_cr_initial_reference}
    q_L := \eta .
\end{equation}
Then, for $m=L,L-1,\ldots,1$, train $\widehat{\bm{A}}_m$ so that
\begin{equation}\label{eqn:sf_cr_stage_push}
    (\widehat{\bm{A}}_m)_{\sharp}\pi
    \approx
    q_m .
\end{equation}
Equivalently, $\widehat{\bm{A}}_m$ is obtained by solving
\begin{gather}\label{eqn:sf_cr_stage_objective}
    \min_{\bm{\alpha}_m}
    \frac{1}{N}
    \sum_{j=1}^{N}
    \left[
        -\log q_m\!\left(
            \widehat{\bm{A}}_m(\bm{x}^{(j)};\bm{\alpha}_m)
        \right)
        -
        \log\det\nabla \widehat{\bm{A}}_m
        (\bm{x}^{(j)};\bm{\alpha}_m)
    \right],\\[3pt]\nonumber
    \partial_k \widehat A_{m,k}>0,\quad k=1,\ldots,d .
\end{gather}
The map-induced approximation of the target after stage $m$ is
\begin{equation}\label{eqn:sf_cr_induced_density}
    \widehat{\pi}_m
    :=
    \widehat{\bm{A}}_m^{\sharp}q_m
    \approx
    \pi .
\end{equation}
For $m\geq 2$, this density becomes the reference for the next stage:
\begin{equation}\label{eqn:sf_cr_reference_update}
    q_{m-1}
    :=
    \widehat{\pi}_m
    =
    \widehat{\bm{A}}_m^{\sharp}q_m,
    \qquad
    m=L,\ldots,2 .
\end{equation}
The resulting composed map is
\begin{equation}\label{eqn:sf_cr_final_map}
    \bm{S}^{\mathrm{deep\text{-}CR}}
    :=
    \widehat{\bm{A}}_L
    \circ
    \widehat{\bm{A}}_{L-1}
    \circ
    \cdots
    \circ
    \widehat{\bm{A}}_1
    \approx
    \bm{S} .
\end{equation}
Thus, each new factor corrects the discrepancy between the current map-induced approximation and the true target. This construction corresponds to the changing-reference composition described for deep transport maps in \cite{peherstorfer2019transport}.

\subsection{Single-fidelity fixed-reference composition}
\label{app:single_fidelity_fixed_reference}

A second single-fidelity strategy keeps the reference density fixed at every stage. This corresponds to the alternative composition direction discussed in \cite{peherstorfer2019transport}, where the target is repeatedly pulled back through the current map approximation while the reference remains fixed; related fixed-reference constructions appear in data-driven transport and density-estimation methods \cite{tabak2013family,trigila2016data}.

The construction proceeds as follows. Define a cumulative map and train from $m=L$ down to $m=1$. Set
\begin{equation}\label{eqn:sf_fr_initial_composition}
    \widehat{\bm{C}}_{L+1}
    :=
    \mathrm{Id}.
\end{equation}
At stage $m=L,L-1,\ldots,1$, define the transformed density
\begin{equation}\label{eqn:sf_fr_transformed_density}
    q_m
    :=
    (\widehat{\bm{C}}_{m+1})_{\sharp}\pi
\end{equation}
and the corresponding transformed samples
\begin{equation}\label{eqn:sf_fr_transformed_samples}
    \bm{y}_m^{(j)}
    :=
    \widehat{\bm{C}}_{m+1}(\bm{x}^{(j)}),
    \qquad
    j=1,\ldots,N .
\end{equation}
At the first stage, $m=L$, we have $\widehat{\bm{C}}_{L+1}=\mathrm{Id}$, so $q_L=\pi$ and $\bm{y}_L^{(j)}=\bm{x}^{(j)}$. At later stages, $q_m$ is the distribution obtained by applying the cumulative map already learned to the target distribution.

The residual map $\widehat{\bm{B}}_m$ is trained so that
\begin{equation}\label{eqn:sf_fr_stage_push}
    (\widehat{\bm{B}}_m)_{\sharp}q_m
    \approx
    \eta .
\end{equation}
Equivalently, $\widehat{\bm{B}}_m$ is obtained by solving
\begin{gather}\label{eqn:sf_fr_stage_objective}
    \min_{\bm{\beta}_m}
    \frac{1}{N}
    \sum_{j=1}^{N}
    \left[
        -\log \eta\!\left(
            \widehat{\bm{B}}_m(\bm{y}_m^{(j)};\bm{\beta}_m)
        \right)
        -
        \log\det\nabla \widehat{\bm{B}}_m
        (\bm{y}_m^{(j)};\bm{\beta}_m)
    \right],\\[3pt]\nonumber
    \partial_k \widehat B_{m,k}>0,\quad k=1,\ldots,d .
\end{gather}
The cumulative map is updated by composition:
\begin{equation}\label{eqn:sf_fr_composition_update}
    \widehat{\bm{C}}_m
    :=
    \widehat{\bm{B}}_m
    \circ
    \widehat{\bm{C}}_{m+1},
    \qquad
    m=L,L-1,\ldots,1 .
\end{equation}
The final fixed-reference deep map is
\begin{equation}\label{eqn:sf_fr_final_map}
    \bm{S}^{\mathrm{deep\text{-}FR}}
    :=
    \widehat{\bm{C}}_1
    =
    \widehat{\bm{B}}_1
    \circ
    \widehat{\bm{B}}_{2}
    \circ
    \cdots
    \circ
    \widehat{\bm{B}}_{L}
    \approx
    \bm{S} .
\end{equation}

This construction can be viewed as progressive Gaussianization: each stage applies the current cumulative map to the target samples and then learns an additional triangular correction to move the transformed samples closer to $\eta$. Unlike the changing-reference construction, the reference density remains fixed. Therefore, if $\eta$ is chosen to be standard Gaussian, every stage retains the same log-concave, factorized-reference structure used in Section~\ref{sec:inverse_transport}.

\subsection{Relation to the multifidelity hierarchy}
\label{app:deep_maps_relation_to_mf}

The two hierarchical multifidelity formulations in Section~\ref{sec:hierarchical_multifidelity} can be understood as multifidelity extensions of these single-fidelity deep-map strategies.

In the changing-reference multifidelity method of Section~\ref{sec:hcr_formulation}, the reference density changes from stage to stage, as in \eqref{eqn:sf_cr_reference_update}. The difference is that the target samples at stage $\ell$ come from the prescribed fidelity density $\pi_\ell$, rather than from the same target density $\pi$ at every stage. The single-fidelity changing-reference composition repeatedly refines an approximation to one target distribution, whereas the multifidelity version uses the prescribed fidelity hierarchy
\[
    \pi_M \rightarrow \pi_{M-1} \rightarrow \cdots \rightarrow \pi_0
\]
to define the stages.

In the fixed-reference multifidelity method of Section~\ref{sec:hfr_formulation}, the reference density remains $\eta$ at every stage, as in \eqref{eqn:sf_fr_stage_objective}. The difference is again in the data: for $\ell<M$, the transformed samples are obtained by applying the lower-fidelity composition
\[
    \widehat{\bm{B}}_{\ell+1}
    \circ
    \widehat{\bm{B}}_{\ell+2}
    \circ
    \cdots
    \circ
    \widehat{\bm{B}}_M
\]
to samples from $\pi_\ell$, while at $\ell=M$ this composition is the identity. Each new map is therefore a residual correction to the Gaussianization achieved by the lower-fidelity composition.

\section{Fixed-Reference Separability for Factorized References}
\label{app:fixed_reference_separability}

We state the componentwise separability property used in Sections~\ref{sec:hfr_formulation}, \ref{sec:nonhierarchical_training}, and \ref{sec:conditions_of_effectiveness}. The result is not specific to a particular multifidelity construction; it follows from triangularity and from using a common factorized reference density.

\begin{proposition}[Componentwise separability for fixed-reference objectives]
\label{prop:fixed_reference_separability}
Let $\eta$ be a factorized reference density on $\mathbb{R}^d$,

\begin{equation}
    \eta(\bm{z})
    =
    \prod_{k=1}^d \eta_k(z_k).
\end{equation}
For each index $r\in\mathcal R$, let $\bm{S}_r:\mathbb{R}^d\to\mathbb{R}^d$ be a monotone triangular map with components
\[
    S_{r,k}(\bm{u}_{\leq k};\bm\Theta_k),
    \qquad k=1,\ldots,d,
\]
where $\bm\Theta_k$ denotes the parameters controlling the $k$th components of all maps appearing in the objective. Assume that the parameter sets $\bm\Theta_1,\ldots,\bm\Theta_d$ are disjoint across components, although parameters may be shared across fidelities or maps within the same component $k$.

Consider the fixed-reference empirical objective
\begin{equation}\label{eqn:general_fixed_reference_objective}
    \mathcal J(\bm\Theta)
    =
    \sum_{r\in\mathcal R}
    \frac{a_r}{N_r}
    \sum_{j=1}^{N_r}
    \left[
        -\log \eta\!\left(
            \bm{S}_r(\bm{u}_r^{(j)};\bm\Theta)
        \right)
        -
        \log\det\nabla \bm{S}_r
        (\bm{u}_r^{(j)};\bm\Theta)
    \right]
    +
    \lambda
    \sum_{k=1}^d
    \mathcal R_k(\bm\Theta_k),
\end{equation}
where $a_r>0$ are fixed weights, $\{\bm{u}_r^{(j)}\}_{j=1}^{N_r}$ are fixed training samples, and the regularizer is componentwise separable. Then
\begin{equation}\label{eqn:general_fixed_reference_decomposition}
    \mathcal J(\bm\Theta)
    =
    \sum_{k=1}^d
    \mathcal J_k(\bm\Theta_k),
\end{equation}
with
\begin{equation}\label{eqn:general_fixed_reference_component_objective}
\begin{aligned}
    \mathcal J_k(\bm\Theta_k)
    =
    \sum_{r\in\mathcal R}
    \frac{a_r}{N_r}
    \sum_{j=1}^{N_r}
    \left[
        -\log \eta_k\!\left(
            S_{r,k}(\bm{u}_{r,\leq k}^{(j)};\bm\Theta_k)
        \right)
        -
        \log \partial_k S_{r,k}
        (\bm{u}_{r,\leq k}^{(j)};\bm\Theta_k)
    \right]
    +
    \lambda\mathcal R_k(\bm\Theta_k).
\end{aligned}
\end{equation}
In particular, if $\eta=\mathcal N(0,I)$, then
\begin{equation}\label{eqn:general_fixed_reference_gaussian_component_objective}
\begin{aligned}
    \mathcal J_k(\bm\Theta_k)
    =
    \sum_{r\in\mathcal R}
    \frac{a_r}{N_r}
    \sum_{j=1}^{N_r}
    \left[
        \frac12
        S_{r,k}^2(\bm{u}_{r,\leq k}^{(j)};\bm\Theta_k)
        -
        \log \partial_k S_{r,k}
        (\bm{u}_{r,\leq k}^{(j)};\bm\Theta_k)
    \right]
    +
    \lambda\mathcal R_k(\bm\Theta_k),
\end{aligned}
\end{equation}
up to additive constants independent of $\bm\Theta$.
\end{proposition}

\begin{proof}
Since $\eta$ factorizes,
\begin{equation}
    -\log\eta(\bm{S}_r(\bm{u}))
    =
    \sum_{k=1}^d
    -\log\eta_k(S_{r,k}(\bm{u}_{\leq k})).
\end{equation}
Since $\bm{S}_r$ is triangular and monotone,
\begin{equation}
    \det\nabla \bm{S}_r(\bm{u})
    =
    \prod_{k=1}^d
    \partial_k S_{r,k}(\bm{u}_{\leq k}),
    \qquad
    \partial_k S_{r,k}>0,
\end{equation}
and therefore
\begin{equation}
    \log\det\nabla \bm{S}_r(\bm{u})
    =
    \sum_{k=1}^d
    \log \partial_k S_{r,k}(\bm{u}_{\leq k}).
\end{equation}
Substituting these two identities into \eqref{eqn:general_fixed_reference_objective} gives a sum over components. Because the parameters and regularizer are assumed to be componentwise separable, the objective decomposes as \eqref{eqn:general_fixed_reference_decomposition}--\eqref{eqn:general_fixed_reference_component_objective}. The Gaussian form \eqref{eqn:general_fixed_reference_gaussian_component_objective} follows from $-\log\eta_k(z)=z^2/2+\mathrm{const}$.
\end{proof}

\paragraph{Application to fixed-reference hierarchical transport.}
At stage $\ell$ of the fixed-reference hierarchical method in Section~\ref{sec:hfr_formulation}, the transformed samples $\{\bm{y}_\ell^{(j)}\}_{j=1}^{N_\ell}$ are fixed before optimizing $\widehat{\bm{B}}_\ell$. The stage objective \eqref{eqn:hfr_stage_objective} is therefore a special case of Proposition~\ref{prop:fixed_reference_separability} with a single map,
\[
    \bm{S}_r=\widehat{\bm{B}}_\ell,
    \qquad
    \bm{u}_r^{(j)}=\bm{y}_\ell^{(j)}.
\]
Thus, for a factorized reference density, each fixed-reference hierarchical stage separates across triangular components. In contrast, the changing-reference hierarchical method generally does not retain this property after the first stage, because its references $q_\ell$ are map-induced and need not factorize.

\paragraph{Application to non-hierarchical multifidelity transport.}
The NH-MF objective \eqref{eqn:NH-MF_all_at_once} is also a special case of Proposition~\ref{prop:fixed_reference_separability}. In this case, take
\[
    \mathcal R=\{0,1,\ldots,M\},
\]
with
\[
    \bm{S}_0=\widehat{\bm{S}}^{\mathrm{NH}},
    \qquad
    \bm{u}_0^{(j)}=\bm{x}_0^{(j)},
    \qquad
    a_0=1,
\]
and, for $\ell=1,\ldots,M$,
\[
    \bm{S}_\ell=\widehat{\bm{S}}^{\mathrm{corr}}_\ell,
    \qquad
    \bm{u}_\ell^{(j)}=\bm{x}_\ell^{(j)},
    \qquad
    a_\ell=w_\ell.
\]
The component parameter set $\bm\Theta_k$ contains the high-fidelity shift parameters, the scale coefficients, and the low-fidelity correction parameters associated with component $k$. Hence, when $\eta$ is factorized, the NH-MF objective decomposes as
\[
    \mathcal J(\bm\Theta)
    =
    \sum_{k=1}^d
    \mathcal J_k(\bm\Theta_k),
\]
even though the componentwise objective $\mathcal J_k$ couples high- and low-fidelity information within the same component.

\section{Regularization selection}
\label{app:regularization_cv}

We select regularization parameters by five-fold cross validation separately for each experimental scenario, where a scenario is defined by the map type, parameterization, cost model, and high-fidelity sample size \(N_0\). For each scenario, one realization of the prescribed training data is generated. Only the high-fidelity dataset
\[
    \mathcal X_0 = \{\bm{x}_0^{(j)}\}_{j=1}^{N_0}
\]
is partitioned into folds; the low-fidelity datasets \(\mathcal X_1\) and \(\mathcal X_2\), when present, are kept fixed and used in full for each fold.

For fold \(s=1,\ldots,5\), let \(\mathcal X_0^{\mathrm{train},s}\) and \(\mathcal X_0^{\mathrm{val},s}\) denote the high-fidelity training and validation subsets. For each candidate
\[
    \lambda \in L
    \coloneq
    \{10^{-1},10^{-2},10^{-3},10^{-4},10^{-5},10^{-6}\},
\]
we train the corresponding map on \(\mathcal X_0^{\mathrm{train},s}\), together with the full low-fidelity datasets for multifidelity methods. Low-fidelity pretraining stages are unregularized; the candidate \(\lambda\) is applied only to the single-fidelity map, the final high-fidelity stage of a hierarchical map, or the high-fidelity shift functions of a non-hierarchical map.

Validation is performed in reference space. The held-out samples are pushed forward by the trained map,
\[
    \widehat{\bm{z}}_{\lambda,s}^{(j)}
    =
    \widehat{\bm{S}}_{\lambda,s}(\bm{x}_0^{(j)}),
    \qquad
    \bm{x}_0^{(j)}
    \in
    \mathcal X_0^{\mathrm{val},s},
\]
and compared with independent samples from \(\eta = \mathcal N(0,I)\). The fold score is the empirical squared maximum mean discrepancy with a Gaussian kernel, averaged over bandwidths \(\{0.5,1,2,4\}\). The chosen regularization constant is the one that satisfies:
\[
    \lambda^\star
    =
    \arg\min_{\lambda \in L}
    \mathrm{CV}(\lambda), \qquad \mathrm{CV}(\lambda)
    =
    \frac{1}{5}
    \sum_{s=1}^{5}
    \mathrm{MMD}_{s}^{2}(\lambda).
\]
This value is then fixed for all reported training replications within the same scenario.

\section{Additional details for Section~\ref{sec:banana_example}}
\label{app:banana_example}
Here, we record the parameter values and exact high-fidelity KR map used in Section~\ref{sec:banana_example}. The target density family is defined in \eqref{eqn:banana_density_family}. The fidelity specific parameters are

\begin{equation}\label{eqn:banana_parameters}
\begin{array}{c|c|c}
\ell & (a_\ell,b_\ell,d_\ell,e_\ell) & h_\ell(z) \\ \hline
0 & (1,1,0,0) & z^2+0.35z^3+0.18z^4 \\
1 & (1.70,0.80,2.00,1.25) & -0.65z^2+0.10z^3-0.05z^4 \\
2 & (1.25,1.25,-0.85,-0.65) & 0.75z^2+0.20z^3+0.10z^4 .
\end{array}
\end{equation}
For the high-fidelity target, the exact KR map is
\begin{equation}\label{eqn:banana4_true_kr}
\bm{S}_0(x_1,x_2)
=
\begin{bmatrix}
x_1\\
x_2-x_1^2-0.35x_1^3-0.18x_1^4
\end{bmatrix}.
\end{equation}

\section{Additional details for Section~\ref{sec:nealfunnel_example}}
\label{app:nealfunnel_example}
This appendix records the scale and shift functions, coordinate-specific modifiers, fidelity parameters, and exact high-fidelity KR map used in Section~\ref{sec:nealfunnel_example}. Recall that samples from each fidelity are generated by \eqref{eqn:nealfunnel_forward}. The conditional scale and shift functions are

\begin{gather}\label{eqn:nealfunnel_scale_shift}
\tau_{\ell,k}(z_1)
=
s_k\alpha_\ell a_k
\left[
\epsilon_\ell+\mathrm{SoftPlus}(\gamma_\ell g_k z_1+\kappa_\ell+k_k)
\right]^{p_\ell+\Delta p_k},\\[3pt]
\qquad
\xi_{\ell,k}(z_1)=\beta_\ell b_k q_k\tanh(r_k z_1),
\end{gather}
where the coordinate-specific modifiers are
\begin{equation*}\label{eqn:nealfunnel_coord_mods}
\begin{aligned}
s_k &= 1+0.1\sin k,
&
q_k &= (-1)^k,
&
a_k &= 1+0.15\sin k,
&
g_k &= 1+0.20\cos k,
\\
k_k &= 0.15\sin(2k),
&
\Delta p_k &= 0.10\cos k,
&
b_k &= 1+0.25\sin(3k),
&
r_k &= 0.75+0.10\,\frac{k}{d-1}.
\end{aligned}
\end{equation*}
The fidelity parameters are
\begin{equation}\label{eqn:nealfunnel_params}
\begin{aligned}
\pi_0:\;(\mu,\sigma,\alpha,\epsilon,\gamma,\kappa,p,\beta)&=(0,1,1,0.15,1.15,0.25,1.15,0.25),\\
\pi_1:\;(\mu,\sigma,\alpha,\epsilon,\gamma,\kappa,p,\beta)&=(0.9,1.45,1.35,0.30,-0.75,0.20,0.85,-0.20),\\
\pi_2:\;(\mu,\sigma,\alpha,\epsilon,\gamma,\kappa,p,\beta)&=(-0.4,1.15,0.9,0.20,0.95,0.05,1.05,0.18).
\end{aligned}
\end{equation}
For $\pi_0$, the exact KR map is explicit. Let $z_1(x)=(x_1-\mu_0)/\sigma_0=x_1$. Then
\begin{align}\label{eqn:nealfunnel_true_kr}
S^{}_{0,1}(x)&=x_1,\\ \nonumber
S^{}_{0,k}(x)&=
\frac{x_k-\beta_0 b_k q_k\tanh(r_k x_1)}
{s_k\alpha_0 a_k\left[\epsilon_0+\mathrm{SoftPlus}(\gamma_0 g_k x_1+\kappa_0+k_k)\right]^{p_0+\Delta p_k}},
\;\;k=2,\ldots,5.
\end{align}

\section{Additional details for Section~\ref{sec:LV}}
\label{app:lv-details}

We collect the simulator, data-gene\-ration, transport-coordinate, SMC--ABC, and
metric-estimation details used in Section~\ref{sec:LV}. Let $\bm Z(t)=(X(t),Y(t))$ denote the predator-prey state, with $\bm Z(0)=(50,100)$. For positive reaction rates $\bm x=(x_1,x_2,x_3,x_4)$, the reaction network is given in Table~\ref{tab:LV_reaction_network}.
\begin{table}[h!]
    \centering
    \caption{Reaction network for Section~\ref{sec:LV}.}
    \label{tab:LV_reaction_network}
    \begin{tabular}{@{} l c l l@{}}
    \toprule
        $r$ & State update & Propensity & Interpretation \\
    \midrule
        1 & $(X,Y)\mapsto (X+1,Y)$ & $a_1(\bm z;\bm x)=x_1XY$ & predator birth by predation \\
        2 & $(X,Y)\mapsto (X-1,Y)$ & $a_2(\bm z;\bm x)=x_2X$ & predator death \\
        3 & $(X,Y)\mapsto (X,Y+1)$ & $a_3(\bm z;\bm x)=x_3Y$ & prey birth \\
        4 & $(X,Y)\mapsto (X,Y-1)$ & $a_4(\bm z;\bm x)=x_4XY$ & prey death by predation \\
    \bottomrule
    \end{tabular}
\end{table}

We use stoichiometric vectors $\bm\nu_1=(1,0)$, $\bm\nu_2=(-1,0)$,
$\bm\nu_3=(0,1)$, and $\bm\nu_4=(0,-1)$. The high-fidelity process is the continuous-time Markov jump process with generator

\begin{equation*}
(\mathcal L f)(\bm z)
=
\sum_{r=1}^4
a_r(\bm z;\bm x)
\left[
    f(\bm z+\bm\nu_r)-f(\bm z)
\right].
\end{equation*}
We write $u_i=\log x_i$. The physical prior used for data generation, SMC--ABC, and risk evaluation is
\begin{equation*}
    \bm u \sim p
    =
    \bigotimes_{i=1}^4 \mathcal U(a,b),
    \qquad
    a=-5,\quad b=2.
\end{equation*}

For transport training and conditional sampling, we map this bounded support to
latent Gaussian coordinates using the fixed componentwise transform $\bm \xi=\Psi(\bm u)$, where

\begin{equation*}
    \xi_i
    =
    \Phi^{-1}\!\left(\frac{u_i-a}{b-a}\right),
    \qquad
    u_i
    =
    \Psi_i^{-1}(\xi_i)
    =
    a+(b-a)\Phi(\xi_i),
    \qquad i=1,\ldots,4.
\end{equation*}
Here $\Phi$ is the standard normal CDF. This transformation is used only to define the transport coordinates; it does not change the physical prior on $\bm u$. Transport maps are trained on joint variables $\bm w=(\bm y,\bm \xi)\in\mathbb R^{13}$, while posterior samples and all reported risks are evaluated after transforming back to $\bm u$.

All simulators are evaluated on the grid $t_j=0.2j$, $j=0,\ldots,150$, so that
the final time is $T=30$. For a trajectory $\{\bm Z_j=(X_j,Y_j)\}_{j=0}^{150}$, the clean summary statistic is

\begin{equation*}
\mathcal S(\bm Z)
=
\left[
\bar X,\bar Y,
\log(\widehat{\mathrm{Var}}(X)+1),
\log(\widehat{\mathrm{Var}}(Y)+1),
\rho_X(1),\rho_X(2),
\rho_Y(1),\rho_Y(2),
\rho_{XY}
\right]
\in\mathbb R^9,
\end{equation*}
i.e., the population means, log-variances, lag-$\{1,2\}$ autocorrelations, and cross-correlation.

For each fidelity $\ell\in\{0,1,2\}$, samples are generated as follows. Draw $\bm u_\ell^{(j)}\sim p$, set $\bm x_\ell^{(j)}=\exp(\bm u_\ell^{(j)})$, simulate a trajectory at fidelity $\ell$, compute the clean summary $\widetilde{\bm y}_\ell^{(j)}=\mathcal S(\bm Z_\ell^{(j)})$, and add independent Gaussian summary noise:

\begin{equation*}
    \bm y_\ell^{(j)}
    =
    \widetilde{\bm y}_\ell^{(j)}
    +
    \bm\varepsilon_\ell^{(j)},
    \qquad
    \bm\varepsilon_\ell^{(j)}
    \sim
    \mathcal N_9(\bm 0,\Gamma).
\end{equation*}
The covariance is shared across fidelities and estimated from clean high-fidelity summaries:
\begin{equation*}
    \Gamma
    =
    \mathrm{diag}(\bm\sigma_{\mathrm{noise}}^2),
    \qquad
    \sigma_{\mathrm{noise},k}
    =
    \gamma\,
    \widehat{\mathrm{sd}}
    \left\{
        \widetilde y_{0,k}^{(j)}
    \right\}_{j=1}^{N_{\mathrm{pool}}},
    \qquad
    \gamma=0.05,
    \qquad
    N_{\mathrm{pool}}=50{,}000.
\end{equation*}
The corresponding transport-training sample is
\begin{equation*}
    \bm w_\ell^{(j)}
    =
    \left(
        \bm y_\ell^{(j)},
        \Psi(\bm u_\ell^{(j)})
    \right)
    =
    \left(
        \bm y_\ell^{(j)},
        \bm \xi_\ell^{(j)}
    \right)
    \in\mathbb R^{13}.
\end{equation*}
Training sets are independent random subsamples from fidelity-specific pools.
Although the pools were generated using common log-rate draws across fidelities, map training does not require or use paired cross-fidelity data.

The high-fidelity simulator defining $\pi_0$ is Gillespie SSA
\cite{gillespie1977exact}. At state $\bm z$, let $a_0(\bm z;\bm x)=\sum_{r=1}^4 a_r(\bm z;\bm x)$. If $a_0=0$, the chain is terminated and the current state is recorded at all remaining observation times. Otherwise, SSA draws
\begin{equation*}
    \Delta t
    =
    -\frac{\log U_1}{a_0(\bm z;\bm x)},
    \qquad
    U_1\sim\mathcal U(0,1),
\end{equation*}
selects the reaction index according to the normalized propensities, and advances events until each observation time is crossed.

The first lower-fidelity simulator defining $\pi_1$ is fixed-step tau leaping
\cite{gillespie2001approximate} with $\tau=0.01$. From state $\bm z$, over a step of length $\Delta t\leq \tau$, independent reaction counts are drawn as
\begin{equation*}
    K_r
    \sim
    \mathrm{Poisson}
    \left(
        a_r(\bm z;\bm x)\Delta t
    \right),
    \qquad
    \bm z
    \leftarrow
    \bm z+\sum_{r=1}^4 K_r\bm\nu_r.
\end{equation*}
To avoid negative populations, death counts are bounded by the current predator and prey populations before applying the update.

The second lower-fidelity simulator defining $\pi_2$ is the deterministic mean-field ODE \cite{kurtz1970solutions}
\begin{equation*}
    \frac{\mathrm dX}{\mathrm dt}
    =
    x_1XY-x_2X,
    \qquad
    \frac{\mathrm dY}{\mathrm dt}
    =
    x_3Y-x_4XY.
\end{equation*}
It is integrated by fourth-order Runge--Kutta with maximum step size $h=0.02$.

For each observed summary $\bm y_{\mathrm{obs}}$, the SMC--ABC reference posterior is computed in physical log-rate coordinates $\bm u$. The prior is the
uniform prior above, and each proposed parameter is simulated with SSA. The ABC
discrepancy uses componentwise prior-predictive normalization. Specifically, for
each conditioning problem we draw $N_{\mathrm{prior}}=1000$ independent prior
samples, simulate SSA summaries $\bm y^{(r)}$, discard numerical failures, and set

\begin{equation*}
    d_k
    =
    \max
    \left\{
        \widehat{\mathrm{sd}}
        \left(
            y_k^{(r)}
        \right),
        10^{-8}
    \right\},
    \qquad
    k=1,\ldots,9.
\end{equation*}
SMC--ABC is then run on normalized summaries with Gaussian ABC kernel
\begin{equation*}
K_\epsilon(\bm y,\bm y_{\mathrm{obs}})
=
\exp
\left\{
-\frac{1}{2\epsilon^2}
\sum_{k=1}^9
\left(
    \frac{y_k-y_{\mathrm{obs},k}}{d_k}
\right)^2
\right\},
\qquad
\epsilon=0.1.
\end{equation*}
The resulting SMC--ABC samples define $\widehat\pi_{\mathrm{ABC}}(\cdot\mid \bm y_{\mathrm{obs}})$.

All posterior comparisons are performed in physical log-rate coordinates $\bm u$. For each prior-predictive test observation, the transport approximation contributes $1000$ conditional samples obtained by solving in $\bm \xi$-coordinates and applying $\Psi^{-1}$. These are compared with a fixed $1000$-sample subset of the corresponding SMC--ABC posterior.

For MMD, we use bandwidths
\begin{equation*}
    \mathcal H=\{0.5,1,2,4\},
    \qquad
    k_h(\bm u,\bm u')
    =
    \exp
    \left(
        -\frac{\|\bm u-\bm u'\|_2^2}{2h^2}
    \right).
\end{equation*}
Given samples $A=\{\bm u_i^A\}_{i=1}^n$ and $B=\{\bm u_i^B\}_{i=1}^m$, the biased squared MMD estimator is
\begin{align*}
\widehat{\mathrm{MMD}}_{h,\mathrm{b}}^2(A,B)
&=
\frac{1}{n^2}
\sum_{i=1}^n
\sum_{j=1}^n
k_h(\bm u_i^A,\bm u_j^A)
+
\frac{1}{m^2}
\sum_{i=1}^m
\sum_{j=1}^m
k_h(\bm u_i^B,\bm u_j^B)
\\
&\hspace{2em}
-
\frac{2}{nm}
\sum_{i=1}^n
\sum_{j=1}^m
k_h(\bm u_i^A,\bm u_j^B).
\end{align*}
We report the multiscale average
\begin{equation*}
    \widehat{\mathrm{MMD}}_{\mathrm{ms}}^2(A,B)
    =
    |\mathcal H|^{-1}
    \sum_{h\in\mathcal H}
    \widehat{\mathrm{MMD}}_{h,\mathrm{b}}^2(A,B).
\end{equation*}
This biased estimator is nonnegative and is used for all MMD risks and floors
reported in Section~\ref{sec:LV}.

For the F\"orstner metric, let $\widehat C_A$ and $\widehat C_B$ denote the sample covariance matrices of two conditional sample sets in $\bm u$-coordinates. We compute
\begin{equation*}
    d_{\mathrm F}(A,B)
    =
    \left(
        \sum_{i=1}^4
        \log^2 \lambda_i
    \right)^{1/2},
\end{equation*}
where $\{\lambda_i\}_{i=1}^4$ are the generalized eigenvalues of $\widehat C_A v=\lambda \widehat C_B v$.

The empirical risks in Section~\ref{sec:LV} are averaged over $M_{\mathrm{test}}=250$ prior-predictive observations. For the posterior-mean risk, let $\{\bm u_{\bullet,m}^{(q)}\}_{q=1}^{1000}$ be transport-induced posterior samples
for test observation $m$, and define

\begin{equation*}
    \widehat{\bm\mu}_{\bullet,m}
    =
    \frac{1}{1000}
    \sum_{q=1}^{1000}
    \bm u_{\bullet,m}^{(q)}.
\end{equation*}
Thus the empirical posterior-mean risk is
\begin{equation*}
    \widehat{\mathcal R}_{\bullet}^{\mathrm{PM}}
    =
    \frac{1}{M_{\mathrm{test}}}
    \sum_{m=1}^{M_{\mathrm{test}}}
    \left\|
        \bm u_m^\star
        -
        \widehat{\bm\mu}_{\bullet,m}
    \right\|_2^2,
\end{equation*}
where $\bm u_m^\star$ is the true generating log-rate for test observation $m$.
The MMD and F\"orstner risks are computed analogously by averaging the corresponding sample-based discrepancies over the same test observations.

The SMC--ABC risk floors estimate the minimum value of the three risk metrics under the idealized situation in which the map-generated samples are perfect; i.e., they belong to the SMC--ABC reference distribution. Using the same $M_{\mathrm{test}}=250$ prior-predictive observations, we compute the risk metrics on $20$ pairs of independent $1000$ sample subsets from precomputed SMC-ABC pools. Averaging over this ensemble of $20$ yields the reported risk floor.

\end{document}